\documentclass[USenglish,onecolumn]{article}
\usepackage{PRIMEarxiv}
\usepackage[utf8]{inputenc}
\usepackage{centernot}
\usepackage{caption}
\usepackage{subcaption}
\usepackage{graphicx}
\usepackage{tabularx}
\usepackage[title]{appendix}
\usepackage{xcolor}
\usepackage{hyperref}
\usepackage{amsthm}
\usepackage{amsmath}
\usepackage{amssymb}
\usepackage{booktabs}

\newcommand{\ind}{\rotatebox[origin=c]{90}{$\models$}}

\newcommand{\W}{{\mathbf{W}}}
\newcommand{\lX}{{\mathbf{X}}}
\newcommand{\Y}{{\mathbf{Y}}}
\newcommand{\Z}{{\mathbf{Z}}}
\newcommand{\G}{{\mathcal{G}}}
\newcommand{\cG}{{\mathcal{G}}}
\newcommand{\cH}{{\mathcal{H}}}

\newcommand{\cE}{{\mathcal{E}}}
\newcommand{\cB}{{\mathcal{B}}}
\newcommand{\cM}{\mathcal{M}}
\newcommand{\cP}{{\mathcal{P}}}
\newcommand{\cS}{{\mathcal{S}}}
\newcommand{\cT}{{\mathcal{T}}}
\newcommand{\cU}{{\mathcal{U}}}
\newcommand{\cV}{{\mathcal{V}}}
\newcommand{\cW}{{\mathcal{W}}}

\newcommand{\e}{{\mathbf{e}}}
\newcommand{\des}{{\mathrm{des}}}

\newcommand{\co}{\mathrm{co}}
\newcommand{\scc}{\mathrm{sc}}
\newcommand{\mic}{{\mathrm{mic}}}

\newtheorem{lemma}{Lemma}
\newtheorem{proposition}{Proposition}
\newtheorem{remark}{Remark}

\newtheorem{theorem}{Theorem}
\newtheorem{corollary}{Corollary}
\theoremstyle{definition}
\newtheorem{definition}{Definition}
  
\title{Foundations of Causal Discovery on Groups of Variables}

\author{
  Jonas Wahl \\
  Technische Universität Berlin  \\
  DLR Institut f\"ur Datenwissenschaften Jena \\
  \texttt{wahl@tu-berlin.de}
   \And
  Urmi Ninad \\ 
  Technische Universität Berlin  \\
  DLR Institut f\"ur Datenwissenschaften Jena \\
  \texttt{urmi.ninad@tu-berlin.de}
   \And
  Jakob Runge \\ 
  Technische Universität Berlin  \\
  DLR Institut f\"ur Datenwissenschaften Jena \\
  \texttt{runge@tu-berlin.de} \\
}

\begin{document}
\maketitle

  \begin{abstract}
{Discovering causal relationships from observational data is a challenging task that relies on assumptions connecting statistical quantities to graphical or algebraic causal models. In this work, we focus on widely employed assumptions for causal discovery when objects of interest are (multivariate) groups of random variables rather than individual (univariate) random variables, as is the case in a variety of problems in scientific domains such as climate science or neuroscience. If the group-level causal models are derived from partitioning a micro-level model into groups, we explore the relationship between micro and group-level causal discovery assumptions. We investigate the conditions under which assumptions like Causal Faithfulness hold or fail to hold. Our analysis encompasses graphical causal models that contain cycles and bidirected edges. We also discuss grouped time series causal graphs and variants thereof as special cases of our general theoretical framework. Thereby, we aim to provide researchers with a solid theoretical foundation for the development and application of causal discovery methods for variable groups.}
\end{abstract}

\section{Introduction}
Inferring causal relationships from observational data and estimating their strength is an ubiquitous task in many research domains for which a multitude of tools \cite{PearlCausality,spirtes_causation_1993,spirtes_anytime_2001,PetJanSch17,ramsey_adjacency-faithfulness_2006,ShimizuLiNGaM,runge_inferring_2019} have been developed throughout the last decades. While the underlying assumptions on the data generating process differ from method to method, the majority of approaches have in common that the causal objects of interest are one-dimensional random variables. However, in some applications, the relevant causal entities can be multivariate groups of variables, such as spatial regions of measurements, or collections of random variables that together describe or approximate a phenomenon of interest, such as the phase and amplitude of an oscillation. For instance, neuroscientists may be interested in causal interactions between brain regions rather than in interactions between individual neurons \cite{semedo_statistical_2020,perich_rethinking_2020}, while climate scientists would like to improve their understanding of spatio-temporal climate modes that extend across large regions on the globe \cite{Runge17Science,runge2015identifying,wang_three-ocean_2019} and interact across long distances. Similarly, economists may want to approximate the economic activity of a given country by a range of different indicators rather than a single univariate index \cite{costanza_development_2014}. 
\\
At present, domain experts typically address such problems by employing the group mean of a variable group as a stand-in for the group as a whole, or by means of more elaborate standard dimension reduction techniques such as principal component analysis (PCA).  For instance, in climate science, the \emph{El Ni}$\tilde{n}$\emph{o Southern Oscillation (ENSO)} is often represented as either a regional average of sea surface temperatures, or as a principal component in a PCA \cite{timmermann_ninosouthern_2018}.  Unfortunately, if some of the causal processes at hand happen at smaller scale than averages or principal components can capture, relevant causal information may be lost. As an example, the group mean of two variable groups $\W$ and $\Y$ may be conditionally dependent given the group mean of a third group $\Z$ while the groups, considered as a whole, satisfy the conditional independence $\W \ind \Y | \Z$, see e.g. \cite{spirtes_causation_1993,Rubensteinetal17}. Causal inference methods based on conditional independence testing such as the PC algorithm might therefore infer different causal structures depending on whether they use group means or the full variable groups as their basic causal objects. Moreover, the dominant mode of \emph{internal} variability of a variable group $\Y$ as recovered by PCA may not be the causally relevant driver of its effect on another group $\Z$ which could for instance be captured more accurately by a higher order principal component. If only the dominant component is consequently used in a causal analysis, then the causal effect of $\Y$ on $\Z$ may be diluted or disappear completely. A practical example of this, again from climate science, that deals with the effect of ENSO on the \emph{North Atlantic Oscillation (NAO)} can be found in \cite{zhang_impact_2019}.
\\
A second approach to causal discovery for variable groups is to run causal discovery algorithms on the totality of all micro-variables and then deduce group-level relationships from the inferred micro-graph. Such an approach will inevitably need to unravel micro-relations of little interest to the group-level problem at hand. For example, one is typically not interested in causal relations between individual grid locations of satellite measurements of temperature data but between different spatial temperature fields as a whole \cite{runge_inferring_2019}. In addition, to be sound, a micro-level causal discovery method may require strong technical assumptions on micro-relations that are again of no relevance to the between-group interactions and it can quickly become computationally inefficient and statistically frail, see e.g. \cite{wahl_vector_2022} for empirical evidence of this for two variable groups. We will return to causal discovery with dimension reduction and full micro-level causal discovery in the final section of this paper, Section \ref{sec.pitfalls}, where we will discuss their strengths and weaknesses in more detail.
\\
An alternative approach to the group-level causal discovery problem is thus to consider variable groups as a whole as the basic causal entities on which to apply available causal discovery methods, see \cite{ParKas17}. For instance, approaches based on conditional independence testing such as the PC-algorithm do not make any assumptions on the dimensionality of its node variables per se and can still be executed provided that its conditional independence tests are adapted to the multivariate setting \cite{Shah2020,Josse2016,Chatterjee2022,hochsprung_increasing_2023}. However, such constraint-based methods rely on two fundamental assumptions, the \emph{causal Markov property} and \emph{causal faithfulness}, or variants thereof, that now have to be assumed directly on the group-level for the methods to be sound. Thus, the following question arises: if causally interacting micro-variables are partitioned into variable groups, see e.g. Figure \ref{fig.general_graph}, do causal discovery assumptions on the micro-level transfer to the group-level and if not, what else is required for these group-level causal discovery assumptions to be valid?

\begin{figure}[h!]
\centering
\includegraphics[scale=0.3]{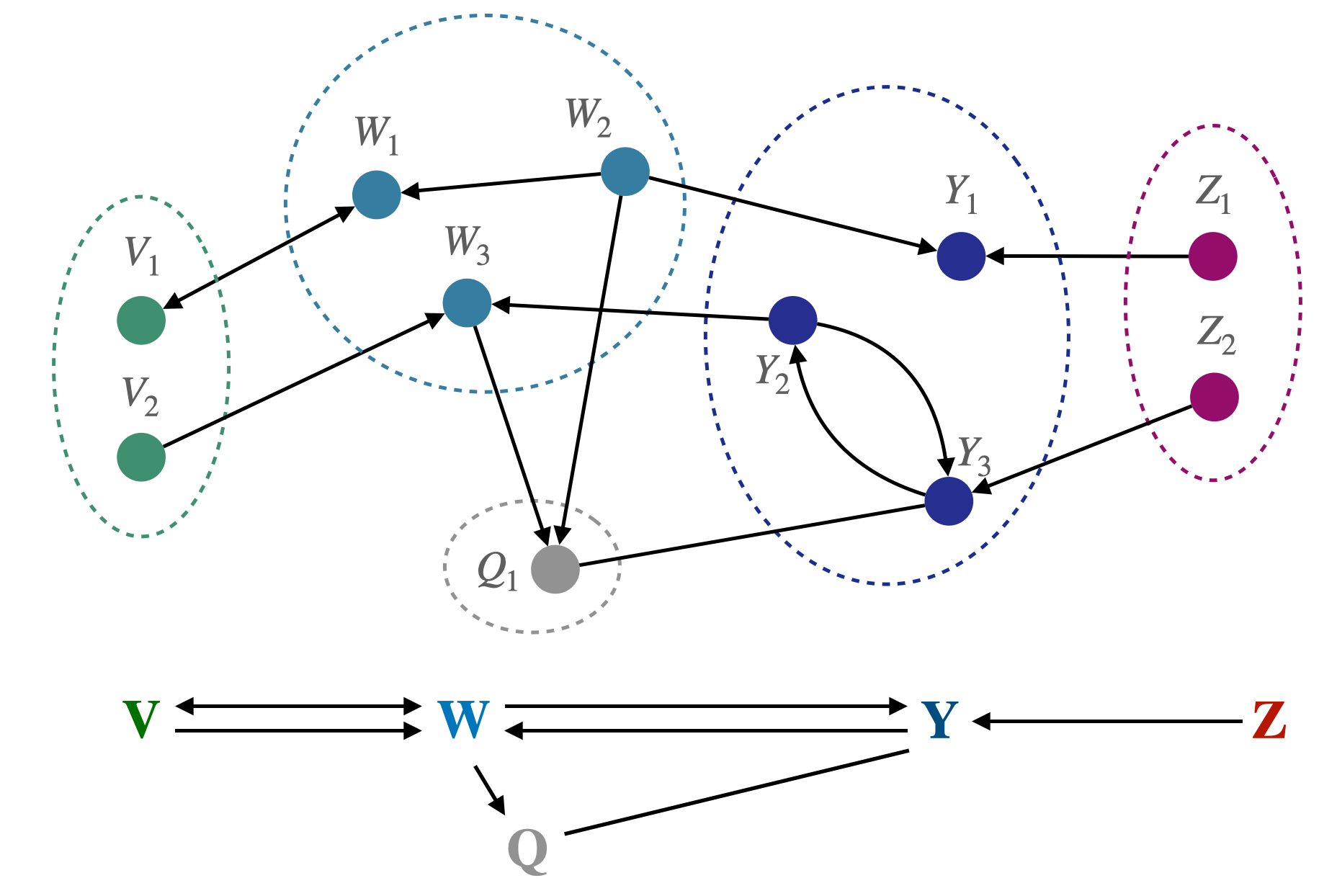}
\caption{A mixed graph over micro-variables is coarsened to a mixed graph of variable groups, see Definition \ref{def.coarsegraph}.In graphical causal modelling, directed edges represent direct causal influences, bidirected edges represent confounding by a hidden variable and undirected edges indicate the presence of a selection variable that has been conditioned on.}\label{fig.general_graph}
\end{figure}

\noindent To answer this question, in this work we provide a thorough theoretical analysis of the relationship between micro- and macro-level causal models with a view on causal discovery assumptions. We do so for causal models that exhibit cyclic as well as acyclic behaviour. Parallel questions on causal effect estimation on directed acyclic graphs over variable groups have been addressed recently in \cite{anand_causal_2023}. The authors of \cite{anand_causal_2023} also present general rules of graphical calculus for acyclic graphs of groups, which we recall and adapt to our setting in Section \ref{sec.groupDMGs} below.
In order to discuss our main results, we now recall that the Markov property and causal faithfulness relate the graphical structure of the model, the \emph{causal graph}, or more precisely its $d-$ or $\sigma-$\emph{separations}, to the observational distribution of the involved variables.
\\
The Markov property states that two variables that can be separated graphically by a separating set $\cS$ are conditionally independent given that set, or for short that $d$- (or $\sigma$-)separation implies conditional independence. The assumption of causal faithfulness on the other hand requires that also the converse implication is true, i.e. that conditional independence implies $d$- (or $\sigma$-)separation. Taken together, both properties thus state that graphical separations and conditional independencies are in exact correspondence to each other. 
\\
While the Markov property is a given in almost every causal inference method, causal faithfulness is more controversial and its validity has been discussed in various places, see e.g. \cite{PearlCausality,weinberger_faithfulness_2018}. As a consequence, weaker versions of causal faithfulness have been developed, most notably \emph{adjacency} and \emph{orientation faithfulness} \cite{ramsey_adjacency-faithfulness_2006}, see also \cite{marx_weaker_2021}. We study under which conditions the causal Markov property, faithfulness and some of its relatives do and do not carry over from a fine grained micro-level causal graphical model to a more coarse grained macro-level graph in which the micro-level variables are partitioned into groups, see Figure \ref{fig.general_graph}. In order to do so we additionally study the relationship between micro-level and coarse grained group-level causal graphs on a purely graphical level, see Section \ref{sec.groupDMGs}.

\noindent As our main results, we show that the Causal Markov property does  transfer from the micro- to the group-level (Theorems \ref{prop.m-Markov_relations} and \ref{prop.sigma-Markov_relations}) relatively straightforwardly, but that this is no longer true for causal faithfulness, a fact that was already noted in empirical simulations in \cite{ParKas17}. We point out that in some sense when dealing with variable groups, the faithfulness assumption is more complicated than was already known: not only does faithfulness fail to transfer to the macro-level, it can even be violated even though its weaker relatives \emph{adjacency} and \emph{orientation faithfulness} \cite{ramsey_adjacency-faithfulness_2006} are both satisfied on the macro-level, see Section \ref{sec.faithfulness} and Figure \ref{fig.non-local}. We are not aware of this type of faithfulness violations (that is, faithfulness being violated but adjacency and orientation faithfulness holding) in other settings and call them \emph{non-local faithfulness violations}.
\\
On the other hand, we also provide two criteria that \emph{do guarantee} macro-level causal faithfulness whenever the variables are \emph{sufficiently well-connected internally}, either through cycles (Theorem \ref{lem.connectivity-criterion}) or through directed or bidirected paths (Theorem \ref{lem.connectivity-criterion2}). This may justify the assumption of causal faithfulness is some settings, as often variable groups are chosen the way the are, exactly because of their internal coherence or their strong internal interactions. Nevertheless, considered in entirety, our discussion shows that faithfulness, already controversial in the univariate case, can be a strong assumption for causal graphs over variable groups and practitioners are advised to proceed with care when assuming it.
\\
We also demonstrate that graphs over variable groups need to be interpreted carefully with respect to their causal meaning as we will discuss in Section \ref{sec.interpretation}. In addition, we point out that the weaker notion of adjacency faithfulness \emph{does transfer} from the micro-level to macro-level (Lemma \ref{lem.adjacencyfaithfulness}). Therefore, when developing causal discovery tools for variable groups, proceeding in line with methods such as the conservative PC-algorithm of \cite{ramsey_adjacency-faithfulness_2006}, that only rely on adjacency faithfulness, may be advisable if there are no domain-specific reasons to believe that faithfulness is a valid assumption.
\\
We end with a discussion on causal discovery for time series, and generalize the widely employed notion of the \emph{time series summary graph}, see e.g. \cite{Runge17Science}, to the notion of \emph{time series summary graphs of groups}. We show that, under a dynamical systems inspired condition that we dub \emph{causal mixing}, stronger causal conclusions can be derived from grouped time series summary graphs. Thus, while causal conclusions on the time-resolved level need to be interpreted carefully, global interactions between groups of processes may be more robust with respect to the standard assumptions of causal inference. To summarize, our main contributions are as follows: \\

\begin{itemize}
\item We extend the theoretical framework of \cite{anand_causal_2023} for graphical causal reasoning between variable groups to $\sigma$-separation and discuss the relationship between micro- and macro-level versions of fundamental graphical properties, such as acyclicity and acyclification (Theorem \ref{thm.coarsen_acyclify}).
\item We discuss Markov properties for $m$- and $\sigma$-separation and show that they transfer from micro- to group-level directed mixed graphs (DMGs), see Theorems \ref{prop.sigma-Markov_relations} and \ref{prop.m-Markov_relations}.
\item We discuss different failure modes of causal faithfulness for graphs over variable groups including an example of a non-local faithfulness violation (Section \ref{sec.faithfulness}).
\item We provide two criteria (Theorems \ref{lem.connectivity-criterion} and \ref{lem.connectivity-criterion2}) that ensure faithfulness on the group-level after coarsening a micro-graph. We also discuss the role of adjacency faithfulness (Lemma \ref{lem.adjacencyfaithfulness}), and an example addressing the applicability of Meek's orientation rules \cite{meek_causal_1995} that was brought forward in \cite{ParKas17} (Subsection \ref{sec.Meek}).
\item We show how time series causal graphs fit into our framework (Section \ref{sec.grouped_ts_graphs}).
\item We elaborate on the difference between apparent and true causation in group DMGs and time series group DMGs (Section \ref{sec.interpretation}).
\item We discuss strengths and failure modes of causal discovery for variable groups through dimension reduction and full micro-level causal discovery and contrast this to an approach that proceeds directly on the group-level (Section \ref{sec.pitfalls}).
\end{itemize}   

We summarize our main results on faithfulness and Markov properties in Figure \ref{fig.summary}. We hope that this work will provide a solid theoretical footing for the development and empirical validation of  group-level causal discovery algorithms in the future.

\begin{figure}[h!]
\centering
\includegraphics[scale=0.38]{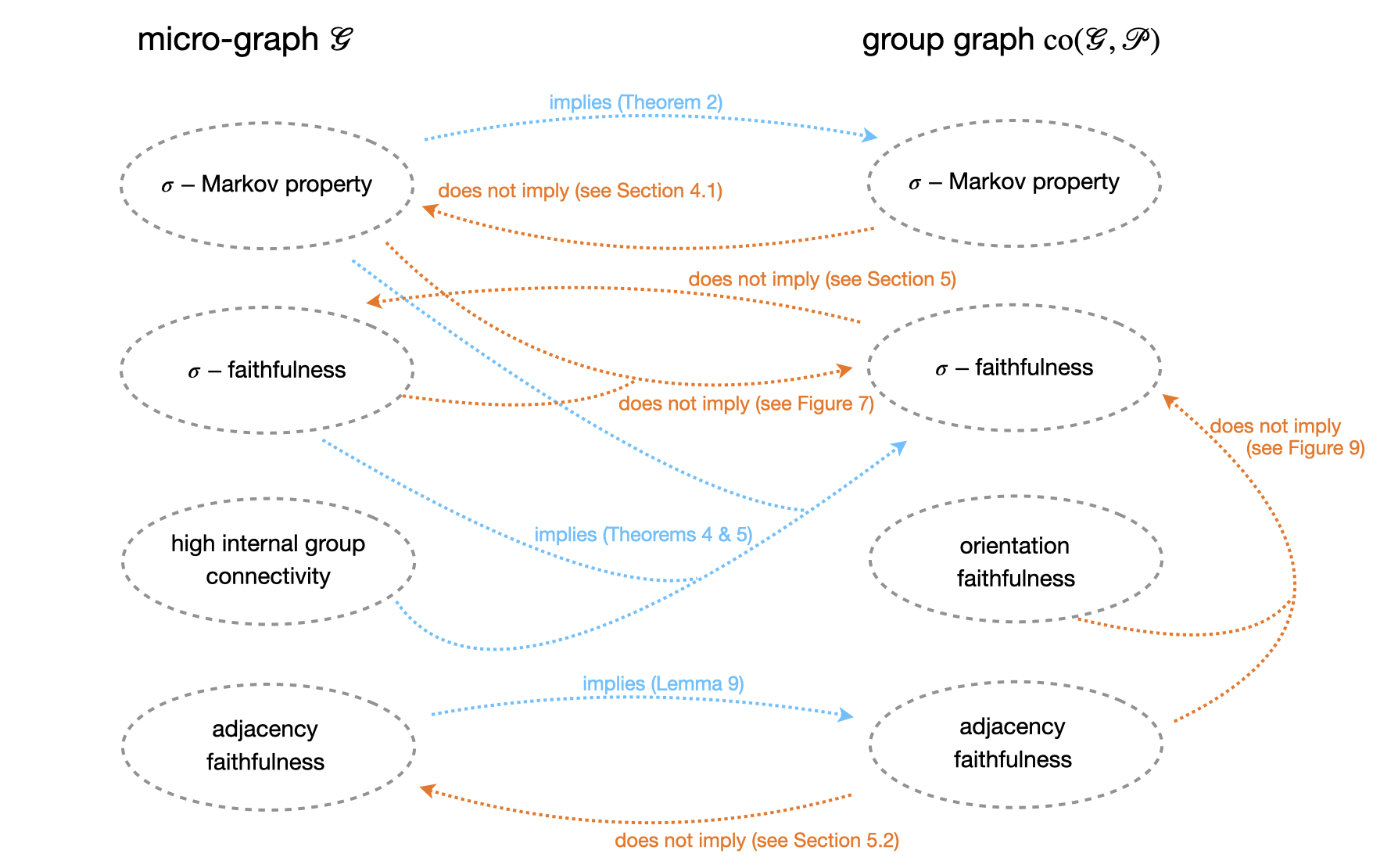}
\caption{A diagrammatic summary of the relationship between the $\sigma$-Markov property and different types of faithfulness on a micro graph $\cG$ and its graph of groups with respect to a partition $\cP$, see Definition \ref{def.graphs} below for details. A blue arrow indicates that all properties from which the arrow emerges imply the target property. An orange arrow indicates that the properties from which the arrow emerges are not sufficient to guarantee the target property.}\label{fig.summary}
\end{figure}
\subsection{Related Work}
The compatibility of averaging across variable and causal inference has been discussed in \cite{Rubensteinetal17} which also provides some toy examples. Arguably, the articles \cite{ParKas17,anand_causal_2023} are closest to our work. The authors of \cite{ParKas17} discussed several causal discovery methods for variable groups, introduced the notion of groupwise faithfulness and provide a first analysis of this property, including some empirical experiments with discrete micro-variables. We expand upon the theoretical analysis of \cite{ParKas17} in several directions, e.g. by including cyclic structures, addressing Markov properties as well as by providing new sufficient criteria for groupwise faithfulness, new examples of faithfulness violations and results on time series. In \cite{ParKas17}, the authors also provide an example in which groupwise faithfulness w.r.t. $d$-separation is deemed insufficient to ensure that the Meek orientation rules \cite{meek_causal_1995}, a fundamental part of the PC-algorithm \cite{spirtes_causation_1993}, still hold. However, we will point out in Section \ref{sec.Meek} that this is no longer true if group-level cycles in the example of \cite{ParKas17} are properly accounted for by replacing $d$-faithfulness with $\sigma$\emph{-faithfulness}. The authors of \cite{anand_causal_2023} present a graphical calculus for $d$-separation over graph of groups, which we will adapt to $\sigma$-separation below, and use this calculus to discuss causal effect estimation for directed acyclic graphs over variable groups, therein called cluster DAGs. The articles \cite{ZscheiJanZhang12} and \cite{wahl_vector_2022} present ways of inferring cause-effect relationships when only two groups of variables are involved. Constraint-based causal discovery methods for variable groups require conditional independence testing for multivariate random vectors which are discussed in various places, e.g. \cite{Shah2020,Josse2016,Chatterjee2022,hochsprung_increasing_2023}. Causal discovery for time series is treated in many works, see e.g. \cite{runge_inferring_2019,runge2023causal,glymour_review_2019} for discussions on state-of-the-art methods.

\section{Preliminaries on (directed) mixed graphs}

To account for latent confounding and selection bias many concepts of causal inference have been extended to mixed graphs \cite{zhang_causal_2008}. Cyclic causal relationships have also been incorporated succesfully into causal graphical modelling \cite{forre_markov_2017} \cite{mooij_constraint-based_2020} \cite{BonFOrPetMoo21} although, for the most part, these works do not deal with undirected edges. 

A \emph{mixed graph} (MG) is a tuple $\G = (\cV,\cE,\cB,\cU)$ of a set of nodes $\cV$, a set of directed edges $\cE$, a set of bidirected edges $\cB$ and a set of undirected edges $\cU$. All these sets are assumed to be countable. Directed edges will be depicted by one-sided arrows $A \rightarrow B $ or $B \leftarrow A$, bidirected edges by two-sided arrows $A \leftrightarrow B$ and undirected edges by simple lines $A-B$. We will assume that graphs considered in this work do not admit \emph{self-edges} of any type, i.e. both nodes of an edge are not allowed to coincide. A \emph{directed mixed graph} (DMG) is a mixed graph without undirected edges, and in this case we will always suppress the (empty) set $\cU$ from the notation. Finally, a \emph{directed graph} (DG) is a directed mixed graph without bidirected edges, and again we will suppress the (empty) set $\cB$ from the notation. A \emph{walk} $\pi$ from $A \in \cV$ to $B \in \cV$ on a mixed graph $\G$ is a finite alternating tuple $ \pi = (\pi(1),e_1,\pi(2),e_2,\dots,e_{m-1},\pi(m)), \ \pi(1) = A, \ \pi(m) = B$ of nodes $\pi(i) \in \cV$ and edges $e_i \in \cE \cup \cB \cup \cU$ such that $e_i$ connects $\pi(i)$ and $\pi(i+1)$, i.e. $e_i \in \{ \pi(i) \rightarrow \pi(i+1), \pi(i) \leftarrow \pi(i+1), \pi(i) \leftrightarrow \pi(i+1), \pi(i) - \pi(i+1)  \}$. A \emph{path} is a walk whose nodes $\pi(1),\dots,\pi(m)$ are all (pairwisely) different. A trivial walk (path) is a walk (path) that consists of only one node and no edges. A walk (path) is called \emph{right-directed} if it is of the form $\pi(1) \to \pi(2) \to \dots \to \pi(m)$, \emph{left-directed} if it is of the form $\pi(1) \leftarrow \pi(2) \leftarrow \dots \leftarrow \pi(m)$ and \emph{directed} if it is left- or right-directed. A \emph{cycle} on $\cG$ is a directed walk  $ \pi = (\pi(1),e_1,\pi(2),e_2,\dots,e_{m-1},\pi(m))$ such that $\pi(1) = \pi(m)$, and a graph is said to be \emph{acyclic} if it does not admit any cycles. As is common practice, \emph{directed acyclic graphs} will be abbreviated as \emph{DAGs}.
A subset of nodes $\cW \subset \cV$ of a mixed graph is \emph{strongly connected} if for any two nodes $A,B\in \cW$ there is a directed path from $A$ to $B$. In particular, there is a cycle between any two nodes in a strongly connected subset. The \emph{strongly connected components} of $\cG$ are the maximal strongly connected subsets of $\cV$, i.e. those that cannot be enlarged without losing their strong connectivity. For any node $A$, the unique strongly connected component that contains $A$ will be written as $\mathrm{sc}(A)$. The strongly connected components of a $\G = (\cV,\cE,\cB)$ form a partition of $\cV$, i.e. $\cV$ is a disjoint union of its strongly connected components. We also use the common conventions that $A \in \cV$ is called a \emph{parent} of $B \in \cV$ if there is a directed edge $A\to B$, and an \emph{ancestor}  of $B$ if there is a directed path from $A$ to $B$. Conversely, in the first case $B$ is called a \emph{child} of $A$, in the latter case $B$ is called a \emph{proper descendant} of $A$. A \emph{descendant} of $A$ is a node that is either $A$ itself or a proper descendant of $A$.
A \emph{collider} of a walk $\pi = (\pi(1),e_1,\pi(2),e_2,\dots,e_{m-1},\pi(m))$ is an inner node $\pi(i), \ 1 <i <m$ of $\pi$ such that both its adjacent edges point into $\pi(i)$. Any inner node of $\pi$ that is not a collider on $\pi$ is consequently called a \emph{non-collider} of $\pi$.

For the purpose of encoding conditional independencies efficiently when modelling causal relationships of random variables graphically, different notions of graphical \emph{separation} have been introduced for different types of graphs. 

\begin{definition}[m-separation, see \cite{zhang_causal_2008}]
Let $\cG =(\cV,\cE,\cB,\cU)$ be a mixed graph and let $\cS \subset \cV$ be a set of nodes. A walk $\pi$ between nodes $A = \pi(1)$ and $B = \pi(m)$ is said to be $m$\emph{-blocked} by $\cS$ if one of the following holds:
\begin{itemize}
\item[(1)] its first node $A$ or its last node $B$ lie in $\cS$;
\item[(2)] there is a collider of $\pi$ that does not have any descendants in $\cS$;
\item[(3)] $\cS$ contains a non-collider of $\pi$.
\end{itemize}
If all walks (or, equivalently, all paths) between $A$ and $B$ are $m$-blocked by $\cS$, we say that $A$ and $B$ are $m$\emph{-separated} by $\cS$ and write $A \bowtie^{m}_{\cG} B | \cS$. If $A$ and $B$ are not $m$-separated by $\cS$, we say that they are $m$\emph{-connected} by $\cS$.
\end{definition}

If the graph $\cG$ is a DAG, $m$-separation is known under the more familiar name $d$-separation. Since $m$-separation can be inadequate to deal with cyclic relationships (see \cite{BonFOrPetMoo21} for a detailed explanation of why this is the case), another type of separation dubbed $\sigma$-separation was introduced in \cite{forre_markov_2017} and studied in \cite{mooij_constraint-based_2020} \cite{BonFOrPetMoo21}. We have only found the definition of $\sigma$-separation for directed mixed graphs in the literature but it is easily adapted to general mixed graphs as well. $\sigma$-separation also reduces to the more familiar notion of $d$\emph{-separation} in the case of directed acyclic graphs.


\begin{definition}[$\sigma$-separation, see \cite{forre_markov_2017}] \label{def.sigma-sep}
Let $\cG=(\cV,\cE,\cB,\cU)$ be a mixed graph and let $\cS$ be a set of nodes.  A walk  $\pi$ from $A = \pi(1)$ and $B = \pi(m)$ is said to be $\sigma$\emph{-blocked} by a subset $\cS \subset \cV$ if one of the following holds:
\begin{itemize}
\item[(1)] its first node $A$ or its last node $B$ lie in $\cS$;
\item[(2)] there is a collider of $\pi$ that does not have any descendants in $\cS$;
\item[(3)] $\cS$ contains a non-collider $\pi(i)$ that has a neighbor $\pi(j), \ j \in \{ i-1,i+1\}$ such that
 \begin{itemize}
 \item  $\pi(j) \notin \mathrm{sc}(\pi(i))$ and
 \item the edge of $\pi$ between $\pi(i)$ and $\pi(j)$ is of the form $\pi(i) \to \pi(j)$ or $\pi(i) - \pi(j)$.
\end{itemize}  
\end{itemize}
If all walks (or, equivalently, all paths) between $A$ and $B$ are $\sigma$-blocked by $\cS$, we say that $A$ and $B$ are $\sigma$\emph{-separated} by $\cS$ and write $A \bowtie^{\sigma}_{\cG} B | \cS$. If $A$ and $B$ are not $\sigma$-separated by $\cS$, we say that they are $\sigma$\emph{-connected} by $\cS$.
\end{definition}

A convenient way of linking the usual notion of $d$-separation on DAGs and $\sigma$-separation is through \emph{acyclification} \cite{BonFOrPetMoo21}.

\begin{definition}[Acyclification of a MG, see \cite{BonFOrPetMoo21}] \label{def.acyclification}
Let $\cG = (\cV, \cE, \cB, \cU)$ be a mixed graph. The \emph{acyclification} of $\cG$ is the graph $\G^{\mathrm{acy}} = (\cV, \cE^{\mathrm{acy}}, \cB^{\mathrm{acy}}, \cU^{\mathrm{acy}})$ defined as follows
\begin{itemize}
\item there is a directed edge $A \to B \in \cE^{\mathrm{acy}}$ if and only if $A \in \mathrm{pa}_{\G}(\mathrm{sc}_{\G}(B)) \backslash\mathrm{sc}_{\G}(B) $;
\item there is an undirected edge $A - B \in \cU^{\mathrm{acy}}$ if and only if $A \notin \mathrm{sc}_{\G}(B) $ and $A-B\in \cU$;
\item there is a bidirected edge $A \leftrightarrow B \in \cB^{\mathrm{acy}}$ if and only if $\mathrm{sc}_{\G}(A) = \mathrm{sc}_{\G}(B)$ or there exist $A' \in \mathrm{sc}_{\G}(A), B' \in \mathrm{sc}_{\G}(B)$ with $A' \leftrightarrow B' \in \cB$.
\end{itemize}
\end{definition}

\begin{figure}[h!]
\centering 
\includegraphics[scale=0.5]{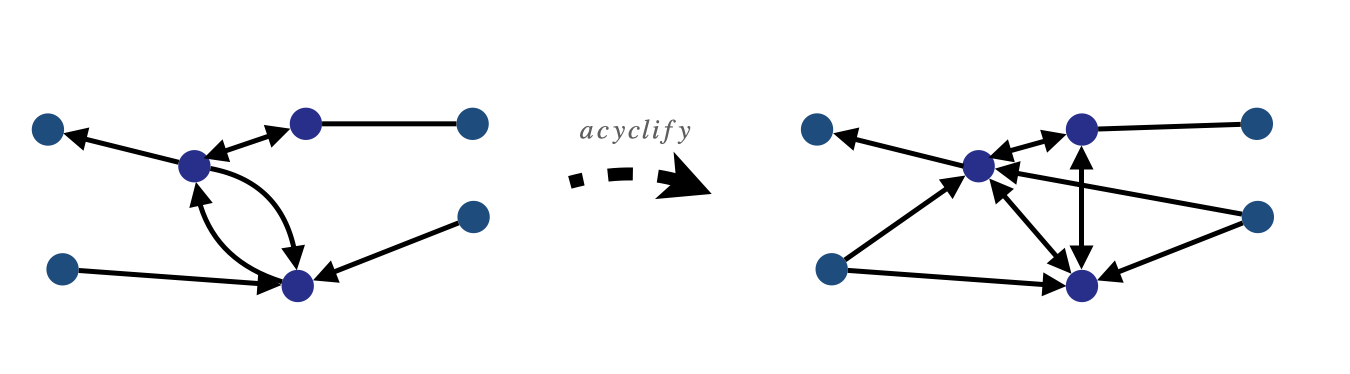}
\caption{Acyclification of a cyclic mixed graph to an acyclic mixed graph.}
\end{figure}

The following result is a straightforward generalization of \cite[Supplement, Proposition A.19]{BonFOrPetMoo21}. It states that $\sigma$-separation on a mixed graph can alternatively be understood as $d$-separation on its acyclification.

\begin{proposition} \label{thm.acyclic-separation}
Let $\cG = (\cV, \cE, \cB, \cU)$ be a mixed graph with acyclification $\G^{\mathrm{acy}}$, let $A,B \in \cV$ and let $\cS \subset \cV$ be a subset of nodes. Then
\[ A \bowtie_{\G}^{\sigma} B \ | \ \cS \qquad \Leftrightarrow \qquad A \bowtie_{\G^{\mathrm{acy}}}^{d} B \ | \ \cS.  \]
\end{proposition}

\section{Group (D)MGs} \label{sec.groupDMGs}
We will now move to the setting where nodes of graphs are no longer supposed to correspond to scalar random variables but to groups of random variables. If the graphs of groups under investigation are assumed acyclic and directed, they appear in the literature under the name \emph{Group DAGs} \cite{ParKas17} or \emph{Cluster DAGs} \cite{anand_causal_2023}. We will adopt the former terminology. Even for directed acyclic graphs, many of the results below including those of Sections \ref{sec.Markov} and \ref{sec.faithfulness} are new. Proofs of the results of this section are either provided immediately or have been moved to Appendix \ref{app.proofs_sec_groupDMGs}.

From now on we will reserve the bold letter $\lX$ for a given countable set $\lX = \{ X_1,X_2,\dots \}$ of \emph{micro nodes}. Although all results of this section are still purely graphical, we will also sometimes freely refer to the micro nodes as \emph{micro-variables} as they will correspond to random variables later on. A \emph{partition} of $\lX$ is a set $\cP$ of pairwise disjoint subsets  of $\lX$ such  that $\cup_{\Y \in \cP} \Y = \lX$. Partitions will always be assumed finite and its elements will be called \emph{variable groups} and will be denoted by bold letters other than $\lX$, e.g. $\W,\Y,\Z$.

\begin{definition}  \label{def.graphs}
Let $\lX$ be the set of micro nodes, and let $\cP$ be a partition of $\lX$ into finitely many subsets. A \emph{(directed) mixed graph of groups} or \emph{group (D)MG} is a (directed) mixed graph $\cG = (\cV, \cE, \cB, \cU)$ whose nodes are the elements of $\cP$, i.e. $\cV = \cP$. If $\cG$ is an acyclic directed graph, we speak of a \emph{Group DAG}.
\end{definition}

To clearly distinguish the usual setting from the group setting, we will speak of a micro (D)MG, respectively a micro DAG if all groups are of size one. There are two natural ways of deriving a Group MG: one can (a) coarsen a MG over micro nodes to a group MG or (b) use a structural causal model over random vectors to induce a group MG directly. The former approach is the main focus of this work while the latter will be defined and shortly discussed in Section \ref{sec.group_SCM}. 

\subsection{From micro-variable graphs to graphs of groups} \label{subsec.micro_to_macro}

If we start out with a mixed graph $\G$ over (the micro nodes in) $\lX$ and a partition  $\cP$ of $\lX$, there is a straightforward way to obtain a group MG over $\cP$ by `coarsening' the graph $\G$. The resulting graph is the quotient of the $\G$ with respect to $\cP$ (in the category-theoretical sense) and is therefore referred to as the \emph{quotient graph} of $\cG$ (w.r.t. $\cP$) in graph theory, see e.g. \cite{mcconnell2005linear}. In the context of causal inference, quotient graphs of (bi)directed graphs were first introduced in \cite[Definition 1]{anand_causal_2023} under the name \emph{cluster DAGs}. 

\begin{definition}[ see \cite{anand_causal_2023}] \label{def.coarsegraph}
Let $\cG$ be a mixed graph over $\lX$, and let $\cP$ be a partition of $\lX$. The \emph{coarse graph} or \emph{quotient graph} $\mathrm{co}(\G,\cP)$ is the mixed graph with nodes $\Y \in \cP$ obtained by
\begin{itemize}
\item[(i)] drawing a directed edge $\Y \to \Z$ if and only if $\Y \neq \Z$ and there is a directed edge $Y \to Z$ on $\G$ with $Y \in \Y$ and $Z \in \Z$;
\item[(ii)] drawing a bidirected edge $\Y \leftrightarrow \Z$ if and only if $\Y \neq \Z$ and there is a bidirected edge $Y \leftrightarrow Z$ on $\G$ with $Y \in \Y$ and $Z \in \Z$;
\item[(iii)] drawing an undirected edge $\Y - \Z,$ if and only if $\Y \neq \Z$ and there is an undirected edge $Y - Z$ on $\G$ with $Y \in \Y$ and $Z \in \Z$.
 \end{itemize}
Note that we do not allow self-edges on $\mathrm{co}(\G,\cP)$ but that multiple edges, each of a different type, are possible between two nodes of $\mathrm{co}(\G,\cP)$. 
\end{definition}


\noindent Clearly, in this generality the newly defined graph $\mathrm{co}(\G,\cP)$ need not be acyclic even if the underlying micro graph $\cG$ is a DAG. On the other hand, the coarse graph \emph{can} be acyclic even if the micro graph $\cG$ does have cycles, see Figure \ref{fig.cycles} for illustrations of both of these statements.

\begin{figure}[h!]
\centering
 \begin{subfigure}[b]{0.4\textwidth}
	 \centering
\includegraphics[scale=0.2]{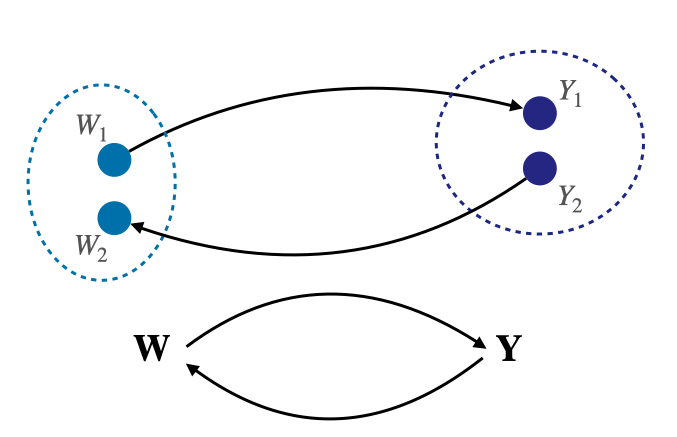}
	\end{subfigure}
     \begin{subfigure}[b]{0.4\textwidth} 
     \centering
\includegraphics[scale=0.2]{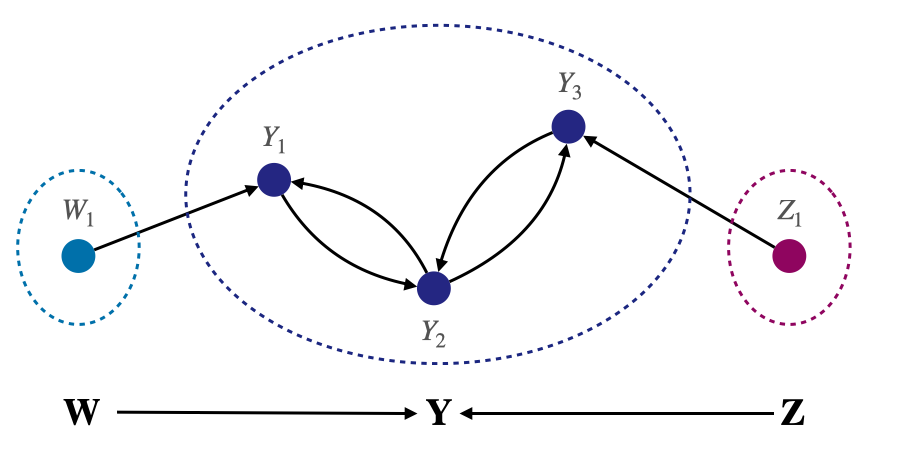}
	\end{subfigure}
\caption{Left: A DAG partitioned such that the resulting group DMG is cyclic, see also \cite[Figure 1(d)]{anand_causal_2023}. Right: A cyclic micro DMG partitioned such that the resulting group DMG is acyclic.} \label{fig.cycles}
\end{figure}

\noindent To discuss separation on directed graphs of groups, it is useful to introduce \emph{walk (path) segments} and \emph{coarse paths}.

\begin{definition} \label{def.part-segments}
Let $\cG$ be a micro MG over $\lX$, and let $\cP$ be a partition of $\lX$. Moreover, let $\pi = (\pi(1),e_1,\pi(2),\dots,e_{m-1},\pi(m))$ be a walk on $\cG$. A subwalk $\pi(i,j) = (\pi(i),e_i,\dots,e_{j-1},\pi(j)), \ i \leq j$ of $\pi$ is called a $\cP$\emph{-segment} of $\pi$ if there exists a group $\Y \in \cP$ such that $\pi(l) \in \Y$ for all $i \leq l \leq j$ and $\pi(i-1), \pi(j+1) \notin \Y$. If $i =1$ or $j = m$, we only require the respective one-sided condition.
\end{definition}

We can thus represent any walk $\pi = (\pi(1),e_1,\pi(2),\dots,e_{m-1},\pi(m))$ on a mixed graph as a sequence $(\pi(i_0,i_1),e_{i_1},\pi(i_1,i_2),e_{i_2},\dots,\pi(i_{s-1},i_s)), \ i_0 = 1, \ i_s = m$ where $\pi(i_l,i_{l+1})$ are the $\cP$-segments of $\pi$ and $e_{i_l}$ are edges that connect nodes that belong to different groups of $\cP$. We call this representation the $\cP$-\emph{segment representation} of $\pi$, see Figure \ref{fig.path-coarsening}.

\begin{definition} \label{def.coarsepath}
Let $\G$ be a micro MG over $\lX$, let $\cP$ be a partition of $\lX$ and let $\Y,\Z \in \cP$. Consider a walk $\pi$ from $Y \in \Y $ to $Z \in \Z$ on $\G$ with $\cP$-segment representation  
$(\pi(i_0,i_1),e_{i_1},\pi(i_1,i_2),e_{i_2},\dots,\pi(i_{s-1},i_s))$, $i_0 = 1, \ i_s = m$. The \emph{coarse walk (path)} $\co(\pi) = (\co(\pi)(1),\tilde{e}_1,\dots,\tilde{e}_{u-1},\co(\pi)(u))$ of $\pi$ is the walk on $\mathrm{co}(\cG,\cP)$ defined as follows:
\begin{itemize}
\item[(i)] $\co(\pi)(l)$ is the unique $\W \in \cP$ containing the nodes of the $\cP$-segment $\pi(i_{l-1},i_l)$;
\item[(ii)] $\tilde{e}_l$ connects $\co(\pi)(l)$ and $\co(\pi)(l+1)$ and is of the same type (directed, bidirected, undirected) as $e_{i_l}.$
\end{itemize}
\end{definition}

\begin{figure}[h!]
\centering
\includegraphics[scale=0.2]{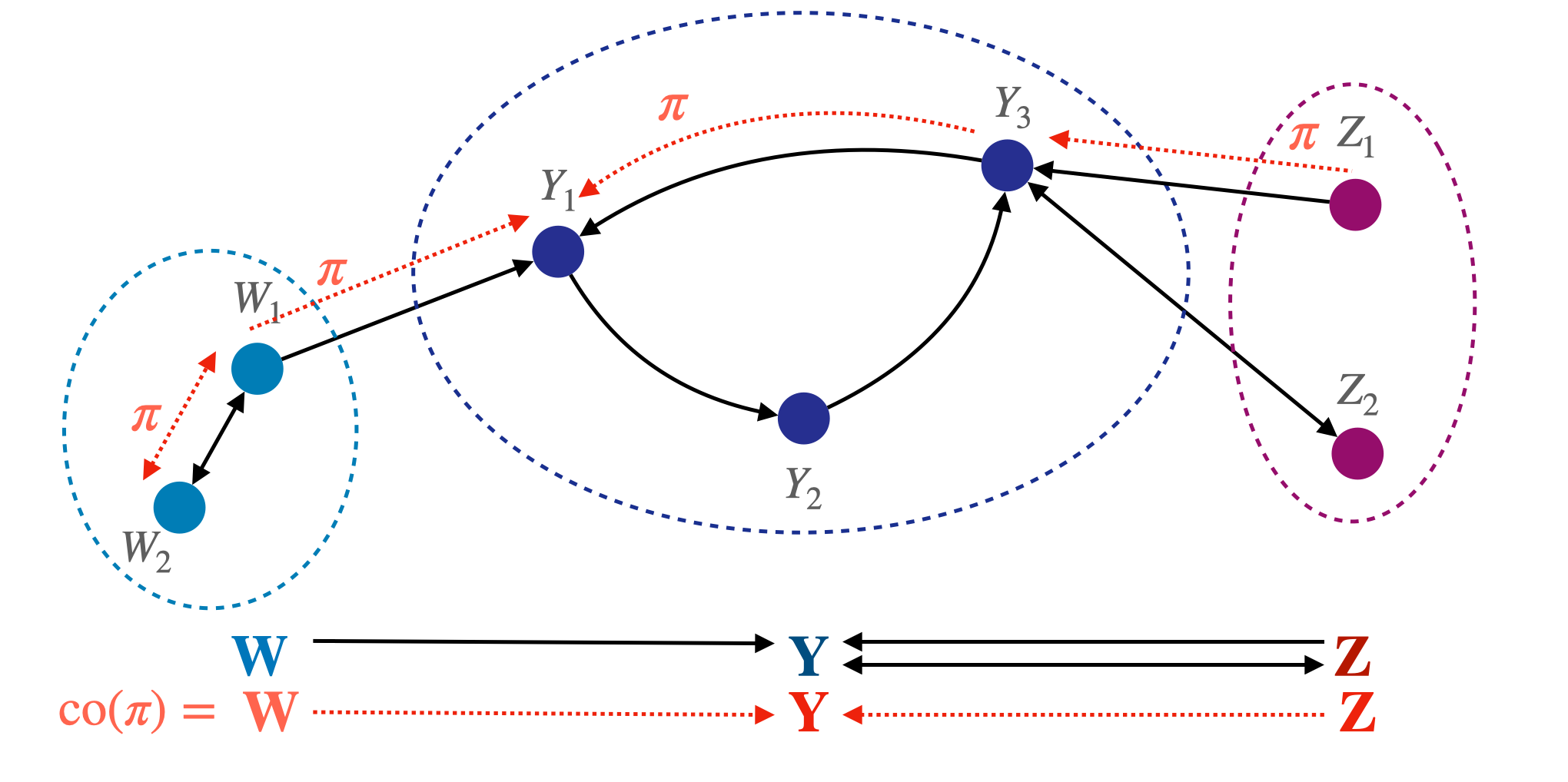}
\caption{The micro path $\pi$ (in red) from $W_2$ to $Z_1$ is coarsened to the path $\co(\pi)$ in the group DMG $\co(\G,\cP), \ \cP = \{ \W,\Y,\Z \}$. The three $\cP$-segments of $\pi$ are $W_2 \leftrightarrow W_1$, $Y_1 \leftarrow Y_3$ and $Z_1$.} \label{fig.path-coarsening}
\end{figure}

\begin{remark}
If $\pi$ in Definition \ref{def.coarsepath} is a path, then $\co(\pi)$ need not be a path as well. For instance if $\pi$ is of the form $\pi = W_1\to Y_1 \to W_2 \to Y_2$ and the micro nodes are grouped as $\W = \{ W_1,W_2 \}, \ \Y = \{ Y_1,Y_2 \}$, then $\co(\pi) = \W \to \Y \to \W \to \Y$ is no longer a path. On the other hand, a micro-walk that is not a path can coarsen to a macro-path if micro-node repetitions only happen within $\cP$-segments. Note also, that if $\pi$ is a directed walk, then $\co(\pi)$ is directed as well, see  \cite[Supplement, Proposition 2]{anand_causal_2023}.
\end{remark}

\begin{lemma} \label{lem.acyclicparts}
Let $\cG$ be a micro MG over $\lX$, and let $\cP$ be a partition of $\lX$.
\begin{itemize}
\item[(i)] If $\mathrm{co}(\G,\cP)$ is acyclic, then for any strongly connected component $\cW$ of $\cG$, there is $\Y \in \cP $ such that $\cW \subset \Y$. 
\item[(ii)] The converse of (i) is not true.
\item[(iii)] If the elements of $\cP$ are exactly the strongly connected components of $\cG$, then $\mathrm{co}(\G,\cP)$ is acyclic.\footnote{ In this case, $\mathrm{co}(\G,\cP)$ is also sometimes referred to as the \emph{condensation} of $\cG$ in the graph theory literature.}
\end{itemize}
\end{lemma}

\begin{definition} \label{def.acyclicity}
We will call a partition $\cP$ of $\lX$ 
\begin{itemize}
\item[(i)] \emph{acyclic} w.r.t. the micro MG $\cG$ if the coarse graph $\co(\G,\cP)$ is acyclic.
\item[(ii)]  \emph{maximally acyclic} if $\cP$ is the partition of $\cG$ into its strongly connected components.
\end{itemize}  
\end{definition}

In particular, acyclicity of a partition $\cP$ entails \emph{unidirectionality}, that is, all directed edges between micro nodes $Y \in \Y$ and $Z \in \Z, \ \Y \neq \Z$ on the micro graph $\cG$ must point in the same direction, e.g. from the elements of $\Y$ to the elements of $\Z$.

It was pointed out in \cite{anand_causal_2023} that coarsening micro DAGs to group DAGs induces an equivalence relation on the set of DAGs over $\lX$ and this observation carries through when the acyclicity assumption on the micro DAGs is dropped. 

\begin{definition} \label{def.equivalence_micro_graphs}
Given a partition $\cP$, we will call two micro MGs $\G$ and $\G'$ $\cP$-equivalent if their coarse graphs with respect to $\cP$ are the same, i.e. if $\mathrm{co}(\G,\cP) = \mathrm{co}(\G',\cP)$.
\end{definition}


The two operations of acyclification in the sense of Definition \ref{def.acyclification} and coarsening in the sense of Definition \ref{def.coarsegraph} do not commute in general, see Figure \ref{fig.acyclify_coarsen}. However, if the partition for coarsening is acyclic with respect to the micro MG, then acyclification of the micro MG has no effect on coarsening.

\begin{figure}[h!]
\centering
\includegraphics[scale=0.3]{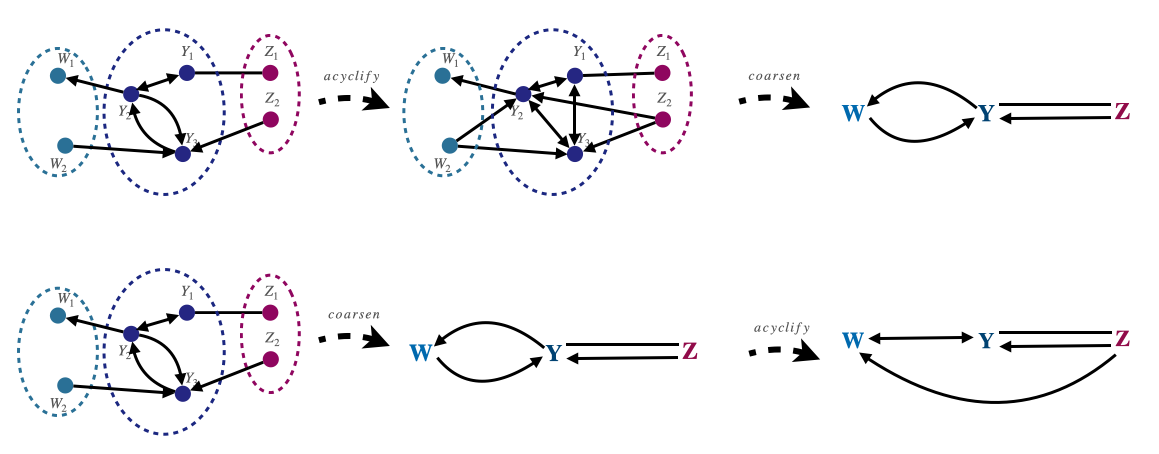}
\caption{Illustration of acyclification and coarsening. In general, these operations do not commute with each other.} \label{fig.acyclify_coarsen}
\end{figure}

\begin{theorem} \label{thm.coarsen_acyclify}
Let $\cG$ be a mixed graph and let $\cP$ be a partition of its nodes. If $\cP$ is acyclic with respect to $\cG$, then
\[
\mathrm{co}(\cG^{\mathrm{acy}},\cP) = \mathrm{co}(\cG,\cP).
\]
\end{theorem}


\begin{lemma} \label{lem.sc-inherit}
Let $\cG$ be a mixed graph and let $\cP$ be a partition of its nodes. If there exist $Y \in \Y$ and $Z \in \Z$ such that $\mathrm{sc}_{\G}(Y) = \mathrm{sc}_{\G}(Z)$, then $\mathrm{sc}_{\mathrm{co}(\G,\cP)}(\Y) = \mathrm{sc}_{\mathrm{co}(\G,\cP)}(\Z)$.
\end{lemma}

\begin{proof}
This result directly follows from the following fact: if there is a directed path from $Y$ to $Z$ (respectively $Z$ to $Y$), then the induced coarse path is a directed path from $\Y$ to $\Z$ (respectively from $\Z$ to $\Y$). 
\end{proof}

The following result clarifies the relationship between $\sigma$-separation on the micro- and the group-level.  It generalizes \cite[Theorem 1]{anand_causal_2023} to $\sigma$-separation in directed mixed graphs and also demonstrates that said theorem does not generalize to arbitrary mixed graphs in which undirected edges are present.

\begin{lemma} \label{lem.d-sep}
Let $\G = (\cV,\cE,\cB)$ be a directed mixed graph, and let $\cP$ be a partition of its nodes. Consider a micro walk $\pi$ on $\G$ and denote its induced coarse walk on $\mathrm{co}(\G,\cP)$ by $\co(\pi)$.
\begin{itemize}
    \item[(i)] If $\co(\pi)$ is $\sigma$-blocked by a set $\cS \subset \cP$ of nodes of $\mathrm{co}(\G,\cP)$, then $\pi$ is $\sigma$-blocked by $\mathcal{T} = \bigcup_{\W \in \mathcal{S}} \W$.
    \item[(ii)] The converse of (i) is not true.
    \item[(iii)] If $\mathcal{S} \subset \cP$ is a set of nodes of $\mathrm{co}(\G,\cP)$ that $\sigma$-separates $\Y, \Z$ in $\tilde{\G}$, then $\mathcal{T} = \bigcup_{\W \in \mathcal{S}} \W$ $\sigma$-separates any pair of micro nodes $Y \in \Y, Z \in \Z$ in $\G$.
    \item[(iv)] (i) and (iii) are no longer true in arbitrary mixed graphs.
\end{itemize}
\end{lemma}


 We also record the analogue of Lemma \ref{lem.d-sep} for $m$-separation for the sake of completeness.

\begin{lemma} \label{cor.m-sep}
Lemma \ref{lem.d-sep} remains true if $\sigma$-separation is replaced by $m$-separation.
\end{lemma}

The proofs of Lemma \ref{lem.d-sep} (i)-(iii) and of Lemma \ref{cor.m-sep} only require straightforward adjustments of the proof of \cite[Theorem 1]{anand_causal_2023} to $\sigma$-separation ($m$-separation), and to the fact that we need to deal with walks instead of paths. We include these proofs in Appendix \ref{app.proofs_sec_groupDMGs} for the convenience of the reader.

\section{Markov properties for Group (D)MGs} \label{sec.Markov}

In this section, we will quickly recap the different types of Markov properties that relate $m$-separation, respectively $\sigma$-separation, to conditional independence statements for scalar node variables. Then we will discuss the transferal of Markov properties from micro graphs to graphs of groups under coarsening. The results of this section are thus no longer purely graphical and micro nodes will always correspond to univariate random variables while nodes of group MGs will consequently always correspond to groups of variables respectively random vectors. 

If $\cG$ is a mixed graph over a set of node variables $\lX$ with joint distribution $P_{\lX}$, then we recall that the pair $(\G,P_{\lX})$ is said to have the $m$-Markov property (or to be $m$-Markovian) if every valid $m$-separation statement on $\G$ implies the corresponding conditional independence statement, i.e. for $A,B \in \lX$ and $\cS \subset \lX$
\[ A \bowtie^m_{\G} B \ | \ \cS \qquad \Rightarrow \qquad A \ind B \ | \ \cS. \]
If the converse implication also holds, that is
\[ A \ind B \ | \ \cS \qquad \Rightarrow \qquad A \bowtie^m_{\G} B \ | \ \cS, \]
then $(\G,P_{\lX})$ is said to be $m$-faithful. Similar properties can also be defined for $\sigma$- instead of $m$-separation: 
$(\G,P_{\lX})$ is said to have the $\sigma$-Markov property (or to be $\sigma$-Markovian) if for $A,B \in \lX$ and $\cS \subset \lX$
\[ A \bowtie^{\sigma}_{\G} B \ | \ \cS \qquad \Rightarrow \qquad A \ind B \ | \ \cS. \]
and is $\sigma$-faithful if the converse implication also holds, that is
\[ A \ind B \ | \ \cS \qquad \Rightarrow \qquad A \bowtie^{\sigma}_{\G} B \ | \ \cS. \]
To introduce analogous properties for mixed graphs of groups, the first observation is that there are now two possible notions of conditional independence that can be considered: pairwise conditional independence and mutual conditional independence. For convenience, we will assume that all distributions have positive densities.

\begin{definition}[Mutual and pairwise independence]
Two groups of random variables $\Y = \{Y_1, Y_2, \dots \}$ and $\Z = \{ Z_1, Z_2,\dots \}$ are called
\begin{itemize}
    \item[(i)] \emph{mutually conditionally independent} given a third group $\W = \{W_1,W_2,\dots \}$ (written $\Y \ind \Z | \W $) if  their joint conditional density almost surely factorizes as $p(\mathbf{y}, \mathbf{z} | \mathbf{w}) = p(\mathbf{y} | \mathbf{w}) p (\mathbf{z} | \mathbf{w})$;
    \item[(ii)] \emph{pairwise (conditionally) independent} given a third group $\W = \{W_1,W_2,\dots \}$ (written $\Y \ind^{pw} \Z | \W $) if for all $Y \in \Y$ and all $Z \in \Z$, we have $Y \ind Z | \W$.
\end{itemize}
\end{definition}

The following well-known characterization illustrates the difference between pairwise and mutual independence nicely: for mutual independence to hold, not only pairwise independence but also conditional independencies involving entries of $\Y$ and $\Z$ in the conditioning set are required as the next lemma illustrates. For a proof of the following result, see \cite[Section 4]{dawid_conditional_1979}.

\begin{lemma} \label{lem.mutual-indep-identity}
Consider groups of random variables $\Y$, $\Z$ and $\W$ and let $\Z' \subset \Z$ be a non-empty subset. The following are equivalent:
\begin{itemize}
    \item[(i)] $\Y$ and $\Z$ are mutually conditionally independent given $\W$.
    \item[(ii)] We have $\Y \ind \Z' \ | \ \W$ and $\Y \ind \Z\backslash\Z' \ | \ \Z', \W$.
\end{itemize}
\end{lemma}

\begin{lemma} \label{lem.mutual-indep-identity2}
Consider disjoint groups of random variables $\Y$, $\Z$ and $\W$ and assume that $\Y$ and $\Z$ are finite and non-empty. If for any $Y \in \Y, \ Z \in \Z$ and any subset $\mathcal{M} \subset \Y \cup \Z \backslash \{Y, Z \}$, we have  $Y \ind Z \ | \ \W, \mathcal{M}$, then  $\Y$ and $\Z$ are mutually conditionally independent given $\W$.
\end{lemma}

The situation is more convenient in graphical models in which the $\sigma$-Markov property and $\sigma$-faithfulness hold on the micro-level. In this case mutual and pairwise conditional independence turn out to be the same in the sense of the following lemma.

\begin{lemma} \label{lem.equivalent-independence}
Let $\G$ be a micro DMG over the micro-variables $\lX$ and suppose that the pair $(\G,P_{\lX})$ is $\sigma$-Markovian and $\sigma$-faithful. Let $\cP$ be a partition of $\lX$ with coarse graph $\co(\G,\cP)$ and let $\cS \subset\cP$. Then two variable groups $\Y,\Z \in \cP \backslash\cS$ are conditionally mutually independent given $\cT := \bigcup_{\W \in \cS} \W$  if and only if they are pairwisely conditionally independent given $\cT$.   
\end{lemma}

\subsection{$\sigma$-Markov properties} \label{subsec.sigma-Markov}

For group MGs, we can now introduce the following Markov properties with respect to $\sigma$-separation.

\begin{definition}[$\sigma$-Markov properties] \label{def.sigma-Markov}
Let $\lX$ be a set of scalar random variables with joint distribution $P_{\lX}$ and let $\cP$ be a partition of $\lX$. Let $\G'$ be a mixed graph with node set $\cP$. We say that $(\G',P_{\lX})$ has the 
\begin{itemize}
\item[(i)] $\sigma$\emph{-Markov property} (or is $\sigma$\emph{-Markovian}) if for $\Y, \Z \in \cP$ and $\cS \subset \cP$, we have
\[ \Y \bowtie^{\sigma}_{\G'} \Z \ | \ \cS \qquad \Rightarrow \qquad \Y \ind \Z \ | \ \cS. \]
\item[(ii)] \emph{weak} $\sigma$\emph{-Markov property} (or is \emph{weakly} $\sigma$\emph{-Markovian}) if for $\Y, \Z \in \cP$ and $\cS \subset \cP$, we have
\[ \Y \bowtie^{\sigma}_{\G'} \Z \ | \ \cS \qquad \Rightarrow \qquad \Y \ind^{pw} \Z \ | \ \cS. \]
\end{itemize}
\end{definition}

\begin{remark}
In the previous definition $\cS$ is a set of sets and therefore, to be precise, we should have written $\bigcup_{\W \in \cS} \W$ instead of $\cS$ in the independence statements. However, whenever the context is clear, we prefer to use $\cS$ instead to keep the notation more simple.   
\end{remark}

$\sigma$-Markovianity transfers nicely from the micro to the macro-level. See Appendix \ref{app.proofs_sec_Markov} for the proof of the following theorem.

\begin{theorem} \label{prop.sigma-Markov_relations}
Let $\G$ be a micro DMG over the micro-variables $\lX$ and suppose that the pair $(\G,P_{\lX})$ is $\sigma$-Markovian. Let $\cP$ be a partition of $\lX$ into finite sets, with coarse graph $\co(\G,\cP)$. Then $(\co(\G,\cP),P_{\lX})$ is $\sigma$-Markovian and consequently weakly $\sigma$-Markovian.
\end{theorem}

\begin{remark} \label{rem.Markov}
It is worthwhile to remark here that while being sufficient, the $\sigma$-Markov property of $(\G,P_{\lX})$ is certainly not necessary for the the $\sigma$-Markov property of $(\co(\G,\cP),P_{\lX})$ as the latter does not care about non-Markovianity \emph{strictly within} variable groups. For instance if $\Y = \{ Y_1,Y_2\}, \ \Z = \{ Z_1 \}$ are two variable groups with only one micro-edge $Y_2 \to Z_1$, then the coarse graph $\Y \to \Z$ is $\sigma$-Markovian with respect to any distribution as there are no $\sigma$-separations. In particular, it is $\sigma$-Markovian w.r.t. distributions in which $Y_1$ and $Y_2$ are not independent, that is for distributions that are not  $\sigma$-Markovian on the micro-graph $\G$.
\end{remark}

\subsection{$m$-Markov properties} \label{subsec.m-Markov}
For good measure, we provide analogues of Definition \ref{def.sigma-Markov} and Theorem \ref{prop.sigma-Markov_relations} for $m$-separation.

\begin{definition}[$m$-Markov properties] \label{def.m-Markov}
Let $\lX$ be a set of scalar random variables with joint distribution $P_{\lX}$ and let $\cP$ be a partition of $\lX$. Let $\G'$ be a mixed graph with node set $\cP$. We say that $(\G',P_{\lX})$ has the 
\begin{itemize}
\item[(i)] $m$\emph{-Markov property} (or is $m$\emph{-Markovian}) if for $\Y, \Z \in \cP$ and $\cS \subset \cP \backslash \{ \Y, \Z  \}$, we have
\[ \Y \bowtie^{m}_{\G'} \Z \ | \ \cS \qquad \Rightarrow \qquad \Y \ind \Z \ | \ \cS. \]
\item[(ii)] \emph{weak} $m$\emph{-Markov property} (or is \emph{weakly} $m$\emph{-Markovian}) if for $\Y, \Z \in \cP$ and $\cS \subset \cP$, we have
\[ \Y \bowtie^{m}_{\G'} \Z \ | \ \cS \qquad \Rightarrow \qquad \Y \ind^{pw} \Z \ | \ \cS. \]
\end{itemize}
\end{definition}

The analogue of Theorem \ref{prop.sigma-Markov_relations} is as follows. Again, see Appendix \ref{app.proofs_sec_Markov} for a proof.

\begin{theorem} \label{prop.m-Markov_relations}
Let $\G$ be a micro DMG over $\lX$ and suppose that the pair $(\G,P_{\lX})$ is $m$-Markovian. Let $\cP$ be a partition of $\lX$ into finite sets, with coarse graph $\co(\G,\cP)$.  Then $(\co(\G,\cP),P_{\lX})$ is $m$-Markovian and in particular weakly $m$-Markovian.
\end{theorem}

\section{Types of Faithfulness for Group (D)MGs} \label{sec.faithfulness}

In this section, we discuss how different notions of faithfulness on scalar mixed graphs relate to faithfulness on a coarsened graph. As we will see, faithfulness is often not preserved under coarsening. However, we will  provide sufficient criteria for faithfulness to hold both in the cyclic and the acyclic setting. We discuss when the strong assumptions that are needed to guarantee faithfulness on the macro-level might be realistic and continue with a discussion on weaker notions of faithfulness. Proofs of the results of this section are provided in Appendix \ref{app.proofs_sec_faithfulness}.

As was already mentioned in \cite{ParKas17}, where coarsening a scalar DAGs $\cG$ to a group DAG $\co(\G,\cP)$ by means of a partition $\cP$, $d$-faithfulness, i.e. faithfulness w.r.t. $d$-separation, need not be preserved. Since DAGs are special cases of mixed graphs, and $m$-separation/respective $\sigma$-separation collapse to $d$-separation on DAGs, this conlusion does not change when either of these separations are considered instead.  Figure \ref{fig.faithfulness-violations} shows simple examples of $\sigma$/$m$-faithfulness violations for $(\co(\G,\cP),P_{\lX})$ that occur even if $\sigma$/$m$-Markovianity and $\sigma$/$m$-faithfulness of a the pair $(\G,P_{\lX})$ is assumed. This observation seriously challenges the most naive approach to causal discovery between groups of variables, namely running the standard PC-algorithm with multivariate conditional independence tests or any adaption thereof that relies on the Causal Faithfulness Condition. We also observe that, conversely, $\sigma$/$m$-faithfulness of $(\co(\G,\cP),P_{\lX})$ need not imply $\sigma$/$m$-faithfulness of $(\G,P_{\lX})$. This is because any $\sigma$/$m$-faithfulness violation for $(\G,P_{\lX})$ that is confined within a variable group will not affect $\sigma$/$m$-faithfulness of $(\co(\G,\cP),P_{\lX})$. As a concrete example, if $(\G,P_{\lX})$ is not $\sigma$/$m$-faithful and $\cP$ collects all variables in one group, then  $(\co(\G,\cP),P_{\lX})$ is always $\sigma$/$m$-faithful for the trivial reason that only one node is present. 

\begin{figure}
\begin{subfigure}[b]{0.45\textwidth}
         \centering
         \includegraphics[scale=0.23]{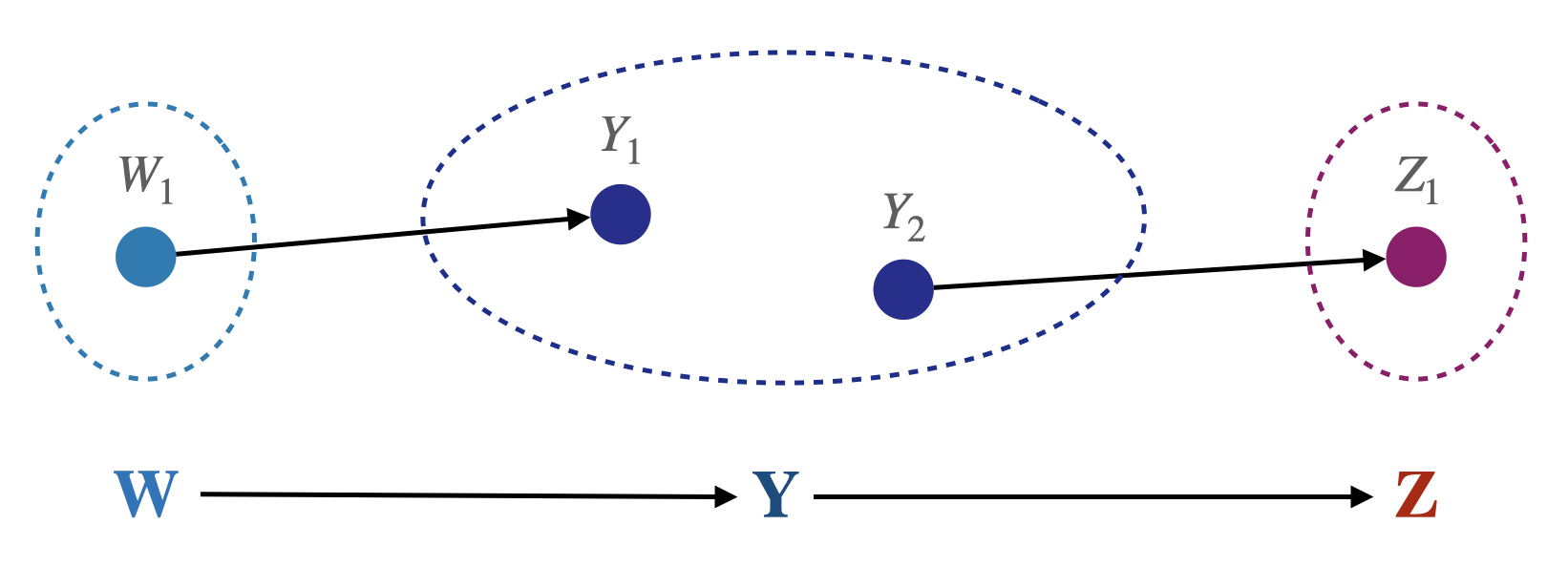}
     \end{subfigure}
     \hfill
     \begin{subfigure}[b]{0.45\textwidth}
         \centering
         \includegraphics[scale=0.22]{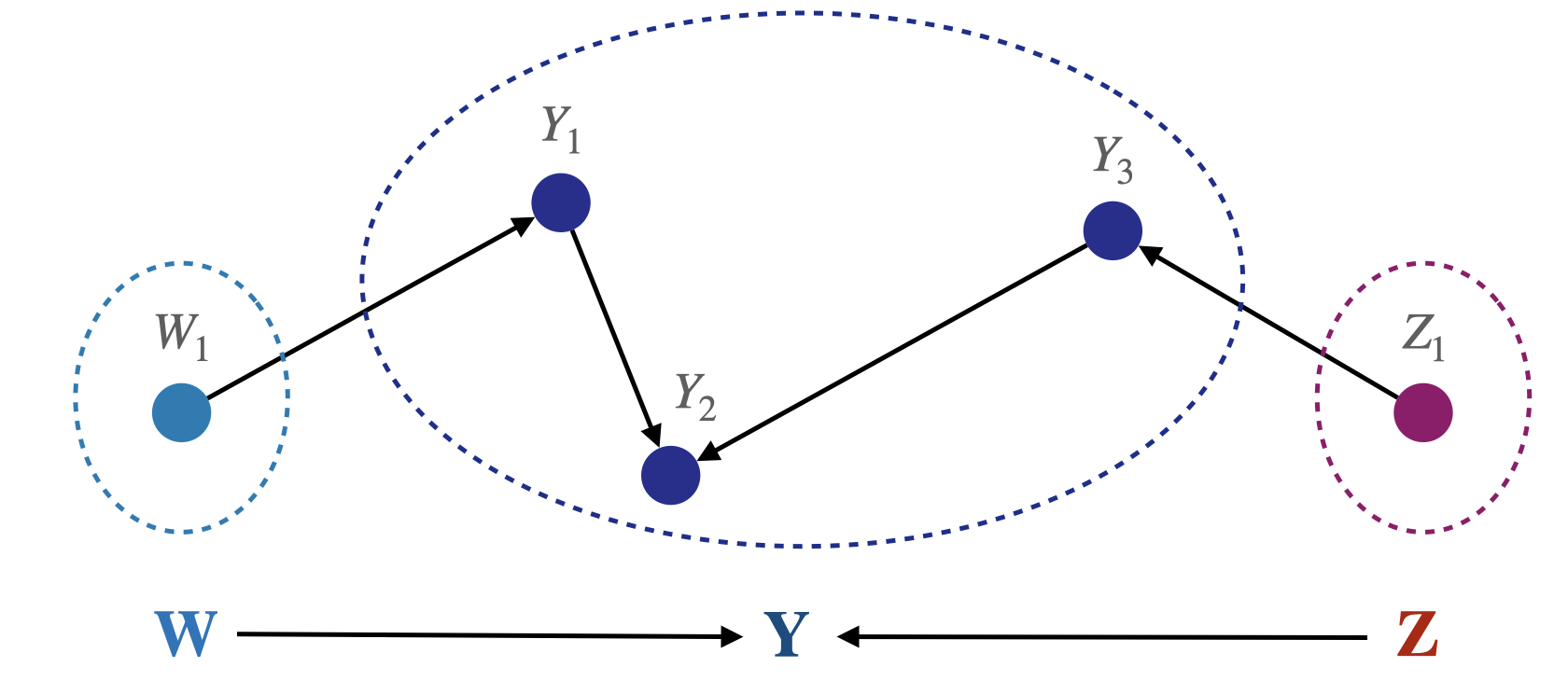}
     \end{subfigure}
\caption{Two simple examples of $d$-faithfulness violations. In the first figure faithfulness is violated due to the internal disconnectedness of $\Y$. In the second figure, conditioning on $\Y$ will open the macro path from $W$ to $\Z$ but closes the micro path, see e.g. Figure 6(i) in \cite[Supplement]{anand_causal_2023}.} \label{fig.faithfulness-violations}
\end{figure}

\subsection{Faithfulness criteria for coarse graphs} \label{sec.faithfulness-criteria}

In this subsection, we will work towards two $\sigma$-faithfulness criteria for group DMGs that are obtained from coarsening a micro-DMG. We will start with the following simple characterization of $\sigma$-faithfulness.

\begin{lemma} \label{lem.macro-micro-faithful}
Let $\G$ be a scalar DMG over the micro-variables $\lX$ and suppose that the pair $(\G,P_{\lX})$ is $\sigma$-Markovian and $\sigma$-faithful. Let $\cP$ be a partition with coarse graph $\co(\G,\cP)$.
Then $(\co(\G,\cP),P_{\lX})$ is $\sigma$-faithful if and only if the following holds: whenever $\Y$ and $\Z$ are $\sigma$-connected by a set $\cS \subset \cP$ then there exist $Y \in \Y$ and $Z \in \Z$ that are $\sigma$-connected by $\mathcal{T} = \bigcup_{\W \in \cS} \W$.
\end{lemma}

\begin{corollary} \label{cor.macro-micro-faithful2}
Let $\G$ be a scalar DMG over the micro-variables $\lX$ and suppose that the pair $(\G,P_{\lX})$ is $\sigma$-Markovian and $\sigma$-faithful. Let $\cP$ be a partition with coarse graph $\co(\G,\cP)$. Assume that for any path $\Pi$ on $\co(\G,\cP)$, there exists a path $\pi$ on $\cG$ such that
\begin{itemize}
\item[(i)] $\co(\pi) = \Pi$ and
\item[(ii)] whenever $\Pi$ is $\sigma$-unblocked by a set $\cS \subset \cP$, then $\pi$ is $\sigma$-unblocked by $\mathcal{T} = \bigcup_{\W \in \cS} \W$.
\end{itemize}
Then $(\co(\G,\cP),P_{\lX})$ is $\sigma$-faithful.
\end{corollary}

\begin{proof}
This is a direct consequence of Lemma \ref{lem.macro-micro-faithful}.
\end{proof}


In Theorem \ref{lem.connectivity-criterion} below, we will now derive a simple sufficient condition that guarantees $\sigma$-faithfulness on a coarsened graph. In a nutshell, it shows that $\sigma$-faithfulness \emph{does hold} if variable groups are sufficiently connected internally. Before formulating Theorem \ref{lem.connectivity-criterion} we need to introduce some additional definitions.

\begin{definition} \label{def.microedges}
Let $\G$ be a mixed graph with edge sets $\cE,\cB,\cU$, and let $\cP$ be a partition of its nodes. Moreover let $\mathbf{e}$ be an edge on $\co(\G,\cP)$.
\begin{itemize}
\item If $\mathbf{e} = \Z \to \Y$ is right-directed, define the set of $\e$-micro edges as 
\[ \mathrm{mic}(\mathrm{e}):= \{ e \in \cE \ ; \ e = Z \to Y \ \mathrm{with} \  Z \in \Z , \ Y \in \Y \}. \]
\item If $\mathbf{e} = \Z \leftarrow \Y$ is left-directed, define the set of $\e$-micro edges as 
\[ \mathrm{mic}(\mathrm{e}):= \{ e \in \cE \ ; \ e = Z \leftarrow Y \ \mathrm{with} \  Z \in \Z , \ Y \in \Y \}. \]
\item If $\mathbf{e} = \Z \leftrightarrow \Y$ is bidirected, define the set of $\e$-micro edges as 
\[ \mathrm{mic}(\mathrm{e}):= \{ e \in \cB \ ; \ e = Z \leftrightarrow Y \ \mathrm{with} \  Z \in \Z , \ Y \in \Y \}. \]
\item If $\mathbf{e} = \Z -\Y$ is undirected, define the set of $\e$-micro edges as 
\[ \mathrm{mic}(\mathrm{e}):= \{ e \in \cU \ ; \ e = Z - Y \ \mathrm{with} \  Z \in \Z , \ Y \in \Y \}. \]
\end{itemize}
Given an arbitrary edge $\e = (\Z,\Y) \in \cE \cup \cB \cup \cU$, the $\e$\emph{-boundary} of $\Z$ is then the projection of $\mic(\e)$ to its source node, i.e. 
\[ \mathrm{bd}_{\e}(\Z) := \{ Z \in \Z \ ; \ \mathrm{there} \ \mathrm{is} \  Y \in \Y \ \mathrm{such} \ \mathrm{that} \ (Z,Y) \in \mic(\e) \} \subset \Z.  \]
Similarly, the $\e$\emph{-boundary} of $\e= (\Z,\Y)$ is defined as 
\[ \mathrm{bd}_{\e}(\Y) := \{ Y \in \Y \ ; \ \mathrm{there} \ \mathrm{is} \  Z \in \Z \ \mathrm{such} \ \mathrm{that} \ (Z,Y) \in \mic(\e) \} \subset \Y.  \]

\begin{figure}[h!]
\centering
\includegraphics[scale=0.3]{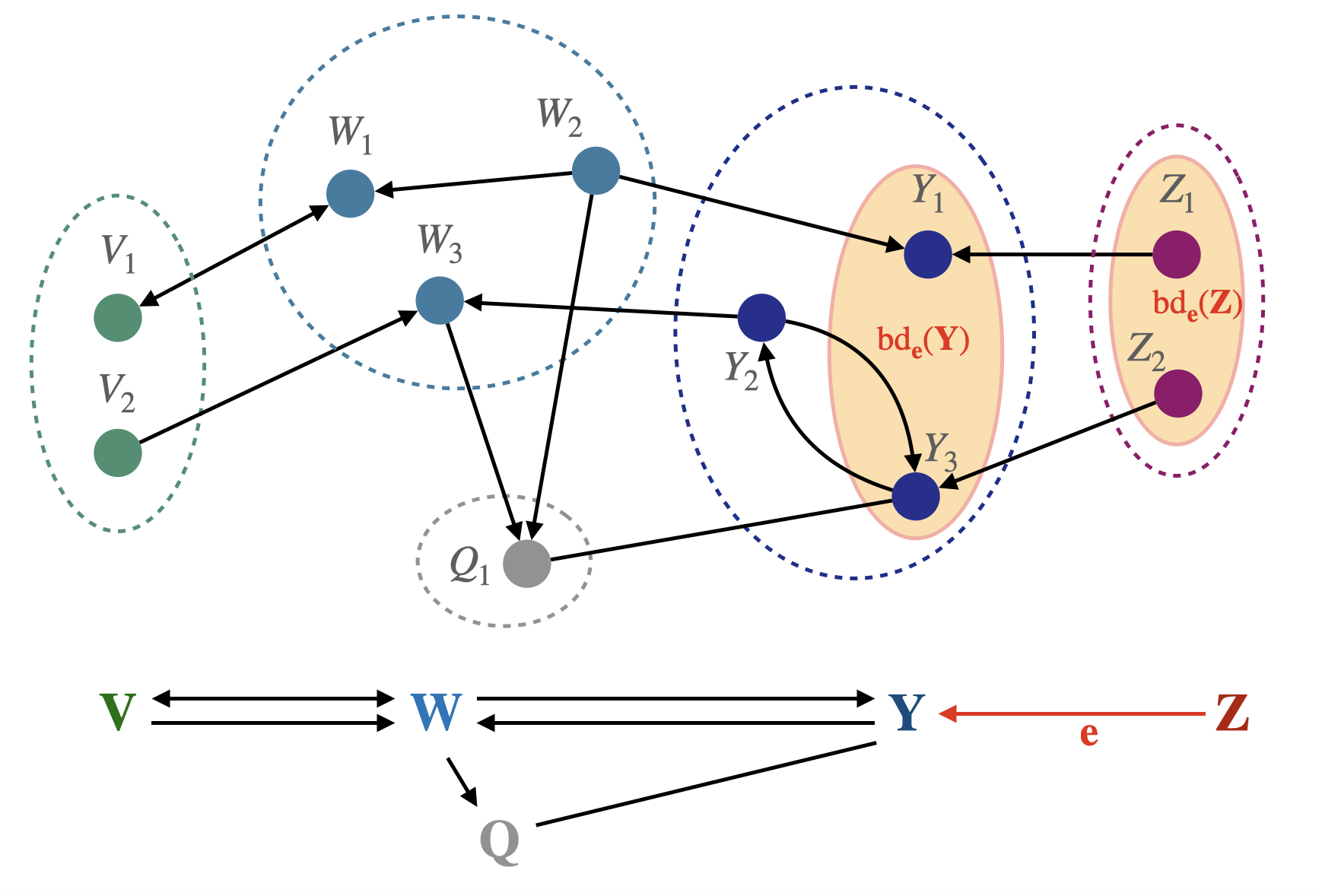}
\caption{The two boundaries of the edge $\mathbf{e}$ that is marked in red in the macro graph.}\label{fig.edge_boundaries}
\end{figure}

\end{definition}

\begin{theorem}[Faithfulness criterion 1] \label{lem.connectivity-criterion}
Let $\G$ be a directed mixed graph over the micro-variables $\lX$ with distribution $P_{\lX}$, and let $\cP $ be a partition of its nodes. Assume the following
\begin{itemize}
\item[(i)] The pair $(\G,P_{\lX})$ is $\sigma$-Markovian and $\sigma$-faithful.
\item[(ii)] For any strongly connected component $\cW$ of $\G$, there exists $\W \in \cP$, with $\cW \subset \W$.
\item[(iii)] For any adjacent pair of edges $\e = (\W,\Y), \e' =  (\Y,\Z)$, and any $Y \in \mathrm{bd}_{\e}(\Y)$ there exists $Y' \in \mathrm{bd}_{\e'}(\Y)$ such that $\mathrm{sc}_{\G}(Y) = \mathrm{sc}_{\G}(Y')$.
\end{itemize}
Then, $(\co(\G,\cP),P_{\lX})$ is $\sigma$-faithful and $\co(\G,\cP)$ is acyclic.
\end{theorem}

\begin{corollary} \label{cor.components-faithful}
Let $\G$ be a directed mixed graph over the micro-variables $\lX$ with distribution $P_{\lX}$, and let $\cP$ be the partition into the strongly connected components of $\G$. If the pair $(\G,P_{\lX})$ is $\sigma$-Markovian and $\sigma$-faithful, then $(\co(\G,\cP),P_{\lX})$ is $\sigma$-faithful.
\end{corollary}

Corollary \ref{cor.components-faithful} is no longer true if $\sigma$-separation is replaced by $m$-separation. The graph on the right of Figure \ref{fig.cycles} provides a counterexample, as every micro-path between groups $\W$ and $\Z$ is $m$-blocked by $\Y$ but $\sigma$-unblocked by $\Y$. This example serves as another illustration that the notions of separation entail different consequences, see \cite{BonFOrPetMoo21} for more. \\

\noindent The previous results show that if cyclic relationships are present \emph{internal} to the variable groups, this can be an advantage for causal discovery rather than a disadvantage. Assuming a variable group to be well connected internally to achieve $\sigma$-faithfulness on the group-level is to some degree at odds with assuming acyclicity on the micrograph as acyclicity disallows paths to be present if they induce a cycle. However, if one zooms in on the proof of Theorem \ref{lem.connectivity-criterion}, it becomes clear that condition (iii) can be replaced by weaker sufficient conditions that still guarantee $\sigma$-faithfulness, even if the micrograph is acyclic. These conditions need to be formulated separately for (almost) mediators, confounders and colliders and are therefore more technical to formulate. Here, by an almost mediator we mean a motive of the form $ A \leftrightarrow B \rightarrow C$ (right-directed almost mediator) or $ A \leftrightarrow B \leftarrow C$ (left-directed almost mediator).  For (almost) mediators, condition (iii) can be replaced by \\
\begin{itemize}
\item[(iii-a)] For any adjacent pair of edges $\e = \W\to \Y, \e' =  \Y\to \Z $ (or $\e = \W\leftrightarrow \Y, \e' =  \Y\to \Z$), and any $Y \in \mathrm{bd}_{\e}(\Y)$ there exists $Y' \in \mathrm{bd}_{\e'}(\Y)$ and a right-directed (possibly trivial) path $Y \to \dots \to \dots \to Y'$ that does not leave $\Y$.
\item[(iii-b)] For any adjacent pair of edges $\e = \W\leftarrow \Y, \e' =  \Y\leftarrow \Z$ or ($\e = \W\leftarrow \Y, \e' =  \Y\leftrightarrow \Z$), and any $Y' \in \mathrm{bd}_{\e'}(\Y)$ there exists $Y \in \mathrm{bd}_{\e}(\Y)$ and a left-directed (possibly trivial) path $Y \leftarrow \dots \leftarrow \dots \leftarrow Y'$ that does not leave $\Y$.\\
\end{itemize}   
For confounders the corresponding condition becomes \\
\begin{itemize}
\item[(iii-c)] For any adjacent pair of edges $\e = \W\leftarrow \Y, \e' =  \Y\to \Z$, and any $Y \in \mathrm{bd}_{\e}(\Y)$ there exists $Y' \in \mathrm{bd}_{\e'}(\Y)$ and a confounding path $Y \leftarrow \dots \leftarrow Y'' \to \dots \to \dots \to Y'$ that does not leave $\Y$. \\
\end{itemize} 

Finding an appropriate condition for colliders is a bit less straightforward, as faithfulness violations may arise by conditioning on a collider $\Y$, e.g. $\W \to \Y \leftarrow \Z$ in such a way that while a micro-collider inside $\Y$ is unblocked, a non-collider in $\Y$ is blocked again, see e.g. the second example in Figure \ref{fig.faithfulness-violations}. In Lemma \ref{lem.connectivity-criterion} this was avoided by enforcing these non-colliders to only point to neighbors in the same strong connected component and condition (iii) in the Definition of $\sigma$-separation, Definition \ref{def.sigma-sep}. The following condition, although strong, will do the job. \\

\begin{itemize}
\item[(iii-d)] For any adjacent pair of colliding edges $\e = (\W,\Y), \e' =  (\Y,\Z)$, and any $ \mathrm{bd}_{\e}(\Y) \cap \mathrm{bd}_{\e'}(\Y) \neq \emptyset$, i.e. there exist colliding edges $(W,Y),(Y,Z)$ with $(W,Y) \in \mic(\e)$ and $(Y,Z) \in \mic(\e')$. \\
\end{itemize} 

Thus, we have the following $\sigma$-faithfulness criterion that is more meaningful when a micro DMG is acyclic, i.e. an ADMG. Note that in this case, condition (ii) of the following theorem is trivially satisfied. In addition, perhaps surprisingly,  it does not enforce the coarse graph $\co(\G,\cP)$ to be acyclic as did Theorem \ref{lem.connectivity-criterion}.

\begin{theorem}[Faithfulness criterion 2]  \label{lem.connectivity-criterion2}
Let $\G$ be a directed mixed graph over the micro-variables $\lX$ with distribution $P_{\lX}$, and let $\cP $ be a partition of its nodes. Assume the following:
\begin{itemize}
\item[(i)] The pair $(\G,P_{\lX})$ is $\sigma$-Markovian and $\sigma$-faithful.
\item[(ii)] For any strongly connected component $\cW$ of $\G$, there exists $\W \in \cP$, with $\cW \subset \W$.
\item[(iii)] An adjacent pair of edges $\e = (\W,\Y), \e' =  (\Y,\Z)$, satisfies the conditions (iii-a), (iii-b), (iii-c), or (iii-d) depending on whether it is a right-directed (almost) mediator, a left-directed almost mediator, a confounder or a collider, respectively.
\end{itemize}
Then, $(\co(\G,\cP),P_{\lX})$ is $\sigma$-faithful.
\end{theorem} 

The discussion in this section also shows the importance of choosing variable groups carefully if one wants to guarantee $\sigma$-faithfulness which may be a non-trivial task in real-world applications. The authors of \cite{ParKas17} tested empirically how often group-level faithfulness would be violated in Erd\"os-R\'enyi random DAGs with groups of small sizes. They found that such violations were likely to appear in sparse graphs but unlikely to appear in dense random graphs. This matches the theoretical results of this section that internally well-connected groups help to ensure group-level faithfulness.

\subsection{Adjacency and Orientation Faithfulness}

\noindent We will therefore consider the two weaker notions of \emph{adjacency faithfulness} and \emph{orientation faithfulness}. The former is at the base of the conservative PC-algorithm \cite{ramsey_adjacency-faithfulness_2006}, and does transfer from the micro-variable to the group-level. 

\begin{definition}[Adjacency faithfulness]
A  pair $(\cG,P)$ of a mixed graph $\G$ and a distribution $P$ over its node variables is \emph{adjacency faithful} if any two nodes $X,Y$ that are independent given some conditioning set $\cS$ are not adjacent, i.e. they do not share an edge.  
\end{definition}

Note that adjacency faithfulness only makes reference to the skeleton of the graph $\G$ and not to any specific type of separation.

\begin{lemma} \label{lem.adjacencyfaithfulness}
Let $\G$ be a mixed graph over the variables $\lX$ with distribution $P_{\lX}$, and let $\cP $ be a partition that induces the coarse graph $\co(\cG,\cP)$. If the pair $(\cG,P_{\lX})$ is adjacency faithful on $\G$, then the pair $(\co(\cG,\cP), P_{\lX})$ is adjacency faithful as well.
\end{lemma}

\begin{proof}
Suppose that $\Y \ind \Z | \cS$ for some $ \cS \subset \cP\backslash\{\Y,\Z\} $. Because mutual conditional independence implies pairwise conditional independence, it follows by adjacency faithfulness on $\G$ that $Y$ and $Z$ do not share an edge for all $Y \in \Y, Z \in \Z$. By definition of $\co(\cG,\cP)$, $\Y$ and $\Z$ do not share an edge.    
\end{proof}

\begin{remark}
Lemma \ref{lem.adjacencyfaithfulness} does not use the full strength of adjacency faithfulness on the micro-level: in fact it suffices to assume that $X$ and $Y$ \emph{that belong to different variable groups} do not share an edge if they are conditionally independent given a conditioning set $ \cS$. In other words: adjacency faithfulness violations \emph{within a group} do not matter for adjacency faithfulness on the macro-level. 
\end{remark}

Combining Theorem \ref{prop.sigma-Markov_relations} with Lemma \ref{lem.adjacencyfaithfulness}, we see that if $(\cG,P_{\lX})$ is a $\sigma$-Markovian and adjacency faithful pair of a DMG $\cG$ and a distribution of micro-variables $P_{\lX}$, then for a given partition $\cP$, the pair $(\co(\cG,\cP), P_{\lX})$ is strongly $\sigma$-Markovian and adjacency faithful as well. If the graph $\co(\cG,\cP)$ is moreover a DAG, these are exactly the assumptions that the conservative PC algorithm of \cite{ramsey_adjacency-faithfulness_2006} requires to be sound. To our knowledge, soundness of conservative PC has not been discussed beyond the acyclic case, but we believe it to hold as well. This is because soundness of the PC algorithm is not affected by allowing cycles and working with $\sigma$-separation as demonstrated in \cite{mooij_constraint-based_2020}. Recall that the conservative PC algorithm takes the observational distribution as an input and outputs a so-called \emph{e-pattern}, see \cite{ramsey_adjacency-faithfulness_2006} for an exact definition.

\begin{corollary} \label{cor.soundness-conservative-PC}
Let $(\cG,P_{\lX})$ be a $\sigma$-Markovian and $\sigma$-faithful pair of a DMG $\cG$ and a distribution of micro-variables $P_{\lX}$. Let $\cP$ be a partition of $\lX$ such that $\co(\cG,\cP)$ is a DAG. Then the conservative PC algorithm with vector-valued (oracle) conditional independence tests is sound for $\co(\cG,\cP)$ in that it outputs an e-pattern that represents $\co(\cG,\cP)$.
\end{corollary}

In an e-pattern, speficific violations of faithfulness, namely violations of \emph{orientation faithfulness} can be singled out and are marked by a $*$. To recap the definition of orientation faithfulness for DAGs, we recall that a triple of nodes $(X,Y,Z)$ in a DAG is called \emph{unshielded} if there is  an edge between $X$ and $Y$ and an edge between $Y$ and $Z$ but none between $X$ and $Z$.
\begin{definition}[Orientation faithfulness] \label{def.orientation-faithfulness}
Let $\G$ be a DAG over a set of variables $\lX$ with distribution $P_{\lX}$. The pair $(\cG,P_{\lX})$
is called \emph{orientation faithful} if for any unshielded triple $(X,Y,Z)$ the following holds.
 \begin{itemize}
     \item[(O1)] If $(X,Y,Z)$ is a collider, then $X$ and $Z$ are dependent given any subset of $\lX\backslash\{X,Z\}$ that contains $Y$;
      \item[(O2)] If $(X,Y,Z)$ is a non-collider, then $X$ and $Z$ are dependent given any subset of $\lX\backslash\{X,Z\}$ that does not contain $Y$;
 \end{itemize}
\end{definition}

For DMGs with potential cycles, orientation faithfulness is more tricky to define, as the absence of an edge between two nodes $X,Y$ does no longer mean that they can be $\sigma$-separated. To deal with this, we will rather introduce the following notion of \emph{local faithfulness} for DMGs which agrees with orientation faithfulness if the graph is a DAG.

\begin{definition}[Local faithfulness]
Let $\G$ be a DMG over a set of variables $\lX$ with distribution $P_{\lX}$. A \emph{local faithfulness violation} is a short path $(X,e_1,Y,e_2,Z)$ such that there exists a set $\cS \subset \lX \backslash \{X,Y,Z \}$ with  $X \ind Z | \cS$ and $X \ind Z| \cS,Y$.
\\

The pair $(\cG,P_{\lX})$ is called \emph{locally faithful} if there are no local faithfulness violations.
\end{definition} 

\begin{lemma}[see \cite{ramsey_adjacency-faithfulness_2006}] \label{lem.locally-orientation}
If G is a DAG, a pair $(\cG,P_{\lX})$ is locally faithful if it is orientation faithful.
\end{lemma}

Examples of faithfulness violations in the literature are typically either violations of adjacency or orientation faithfulness. Figure \ref{fig.non-local} below shows that if the nodes correspond to variable groups, there are faithfulness violations that are non-local. In other words, both orientation and adjacency faithfulness are satisfied, still ($\sigma$- or $d$-)faithfulness is violated. In particular, such non-local violations would not be marked in the output of the conservative PC algorithm.

\begin{figure}[h!]
\centering
\includegraphics[scale=0.45]{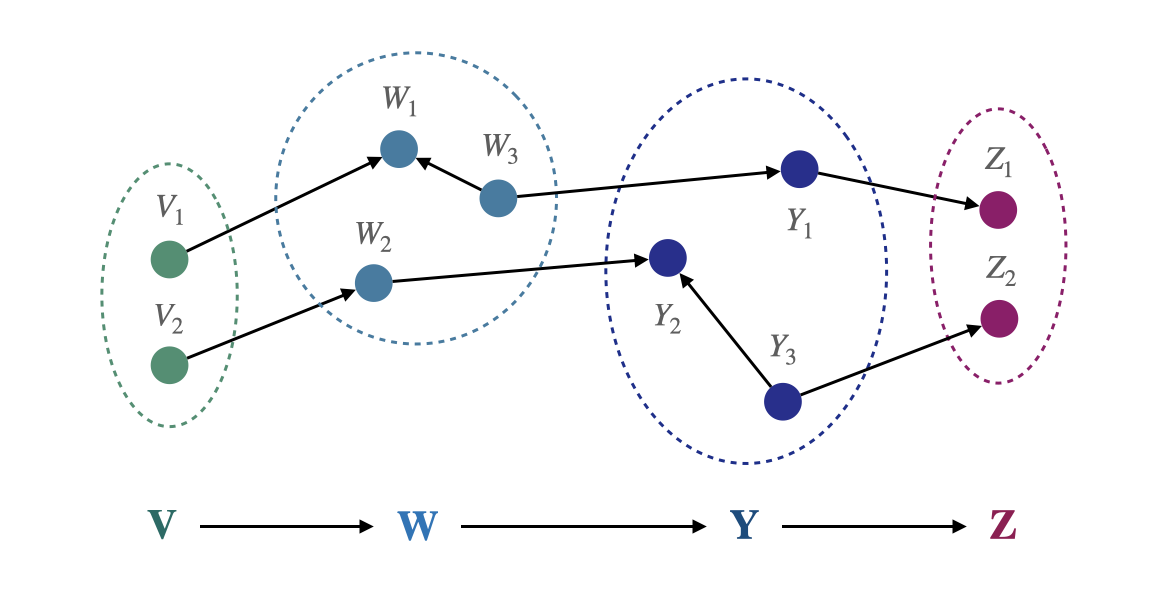}
\caption{An example of a non-local $\sigma$-faithfulness violation (resp. $d$-faithfulness violation as there are no cycles). If the joint distribution is $\sigma$-Markovian and $\sigma$-faithful to the micro graph, the group DAG does not contain any orientation faithfulness violations. At the same time, all micro paths between the groups $\mathbf{V}$ and $\Z$ are $\sigma$-blocked while the only macro path is $\sigma$-open.} \label{fig.non-local}
\end{figure}

\subsection{Faithfulness and Meek's orientation rules revisited} \label{sec.Meek}
Constrained-based algorithms for causal discovery such as the PC-algorithm \cite{spirtes_causation_1993} infer the directionality of arrows in a DAG by first identifying $v$-structures and then applying Meek's orientation rules\footnote{Note that these rules pertain to DAGs, not to general DMGs.} \cite{meek_causal_1995}. In this subsection, only the first of these rules will be relevant. It states that an edge $X - Y$ is to be oriented as $X \to Y$ if there is an edge $Z \to X$ such that $Z$ and $Y$ are non-adjacent. The authors of \cite{ParKas17} discuss the validity of Meek's orientation rules for group DAGs using the example depicted in Figure \ref{fig.PK-example}. Translated to our terminology, their example consists of a micro-variable DAG $\cG$, a partition $\cP = \{\mathbf{V}, \W,\Y,\Z\}$ of the micro-variables and a group DAG $\cG'$ with nodes $\mathbf{V},\W,\Y,\Z$ such that
\begin{itemize}
\item the micro-level pair $(\cG,P)$ is causally Markovian and $d$-faithful, where $P$ is the micro-variable distribution;
\item the macro-level pair $(\cG',P)$ is causally Markovian and $d$-faithful;
\item $\cG' \neq \co(\cG,\cP)$ and in particular $\Y \leftarrow \mathbf{V}$ in $\co(\cG,\cP)$ and $\Y \to \mathbf{V}$ in $\cG'$.
\end{itemize}

\begin{figure}
\begin{subfigure}[b]{0.3\textwidth}
         \centering
         \includegraphics[scale=0.23]{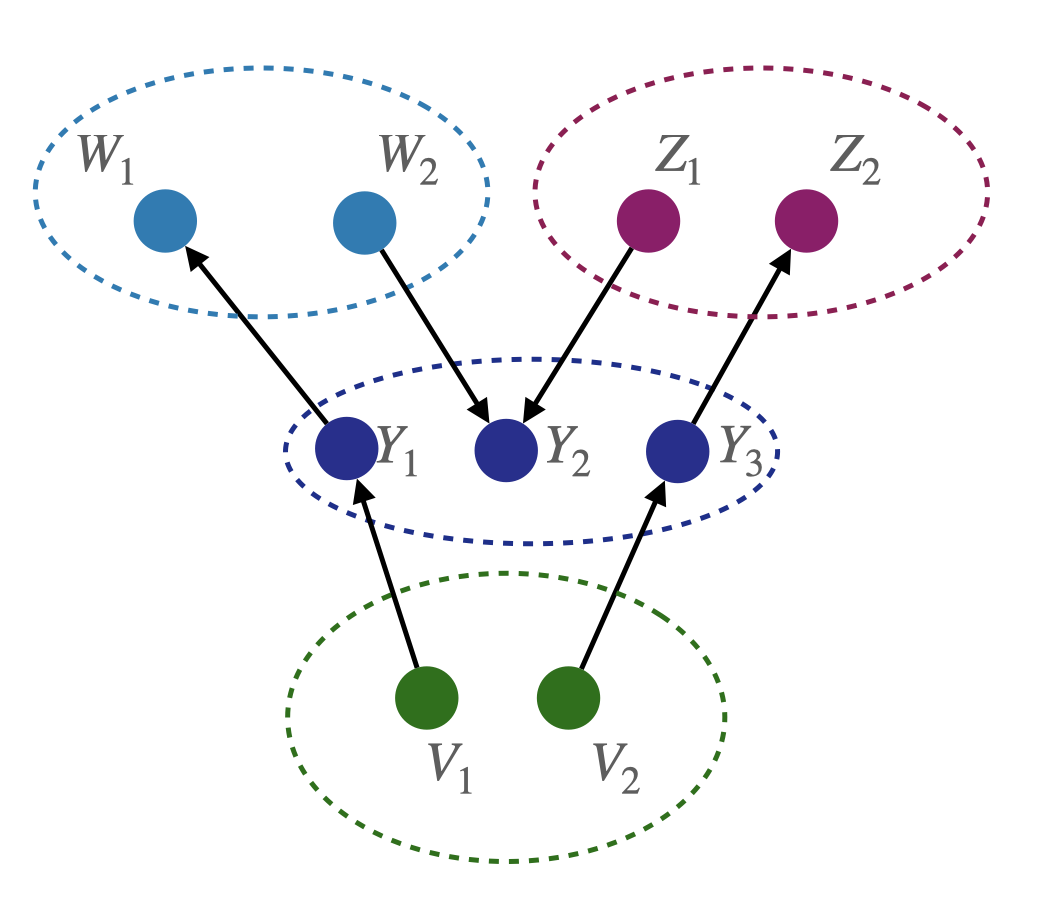}
     \end{subfigure}
     \hfill
     \begin{subfigure}[t]{0.3\textwidth}
         \centering
         \includegraphics[scale=0.31]{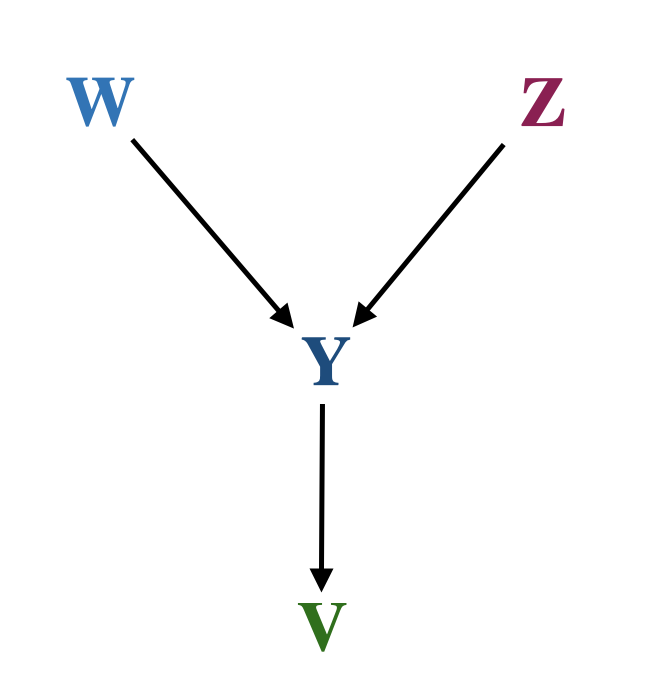}
     \end{subfigure}
     \hfill
     \begin{subfigure}{0.3\textwidth}
         \centering
         \includegraphics[scale=0.3]{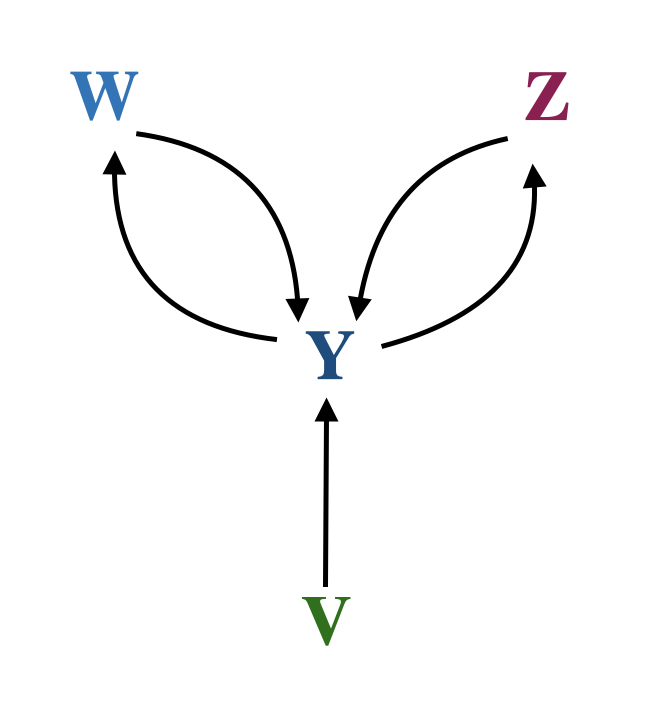}
     \end{subfigure}
\caption{Left: The micro DAG $\cG$ presented in \cite{ParKas17}. Middle: The macro DAG presented in \cite{ParKas17}. Right: The coarse DMG $\co(\cG,\cP)$ with respect to the indicated partition.} \label{fig.PK-example}
\end{figure}

As the above mentioned orientation rule implies the orientation $\Y \to \mathbf{V}$ of $\cG'$ instead of the correct orientation $\Y \leftarrow \mathbf{V}$ in the ground truth group DMG $\co(\cG,\cP)$, the authors of \cite{ParKas17} argue that Meek's orientation rules are no longer valid for group DAGs even if d-faithfulness on the group-level does hold. However, we argue that faithfulness should refer to the cyclic \emph{ground truth graph} $\co(\cG,\cP)$, and the pair $(\co(\cG,\cP),P)$ \emph{does} violate $\sigma$-\emph{faithfulness}: the groups $\W$ and $\Z$ are not $\sigma$-separated in $\co(\cG,\cP)$ but are independent. In fact, by \cite[Corollary 1]{mooij_constraint-based_2020} which does not make assumptions on the dimensionality of the node variables, the PC-algorithm (and thus the Meek rules for DAGs) is sound if the ground truth graph of groups is directed and acyclic, and if this DAG and the joint distribution of the variables are assumed $d$-faithfulness\footnote{Recall that $d$-faithfulness and $\sigma$-faithfulness are equivalent for DAGs.} to each other. To summarize, in the example of \cite{ParKas17}, the Meek rules lead to a wrong orientation, because the graph of groups is incorrectly assumed to be acyclic. 

\section{Grouped Time Series Graphs} \label{sec.grouped_ts_graphs}

When using graphical modes to model causation for time evolving processes, there are several common modeling choices that are discussed in the literature that can all be adapted to the group setting. The arguably most common notion is that of a (stationary) time series DMG (ts-DMG for short) $\cG = (\cV,\cE,\cB)$ in which the processes are unrolled in time and discretized. That is, the processes are modelled as univariate infinite time series $X_i =(X_i(t))_{t \in \mathbb Z}, \ i \in I = \{1,\dots, n\}$ and the nodes of the ts-DMG correspond to the indices $(i,t) \in \cV = I \times \mathbb Z$. As usual, we freely identify an index $(i,t)$ with a variable $X_i(t)$ as long as there is no danger of confusion. In other words, there is a node in the causal graph for every time instance of every process. In addition, directed edges are not allowed to point into the past, i.e. $X_i(s) \to X_j(t)$ implies $s \leq t$. Finally, the stationarity assumption means that the presence of edges only depends on the time lag between nodes and not the actual time instances. More precisely, if there is a directed or bidirected edge $(X_i(s),X_j(t))$, then there is an edge $(X_i(s+u),X_j(t+u))$ of the same type for any $u \in \mathbb Z$. A coarser representation of causal interactions between time series is that of a \emph{time series summary DMG} or \emph{process DMG} $\cG^{\mathrm{sum}} = (\cV^{\mathrm{sum}}, \cE^{\mathrm{sum}},\cB^{\mathrm{sum}})$ in which a node corresponds to a process $X_i$ as a whole, i.e. $\cV^{\mathrm{sum}} = I$. Such graphs thus express whether processes causally influence each other but hold no information on the time lag of the interaction. Depending on the convention, self-edges $(X_i,X_i)$ are allowed or not allowed and we stick to the latter (no self-edges) in this work. While some causal discovery methods \cite{Runge17Science,runge2020discovering,gerhardus_high-recall_2020} aim to infer the time unrolled ts-DMG, others such as Granger causality \cite{granger_investigating_1969} infer the process graph. Clearly, any ts-DMG can be projected to a process DMG by ignoring the time component and adding a (bi)directed edge $(X_i,X_j), \ i \neq j$ if and only if there is a (bi)directed edge $(X_i(s),X_j(t))$ for some $s,t \in \mathbb{Z}$. Note that this is nothing but a special instance of our coarsening operation in the case where micrographs have infinite nodes, see Figure \ref{fig.summary-graph-as-grouped-graph}.

\begin{lemma} \label{lem.summary-graph-coarse}
If $\cG = (\cV,\cE,\cB)$, $\cV = I \times \mathbb Z$ is a time series DMG, then its summary DMG is $\co(\cG,\cP)$ for the partition $\cP = \{ \{X_i\} \times \mathbb{Z} \}_{i \in I} \cong I$.
\end{lemma}

\begin{figure}[h!]
\centering
\includegraphics[scale=0.3]{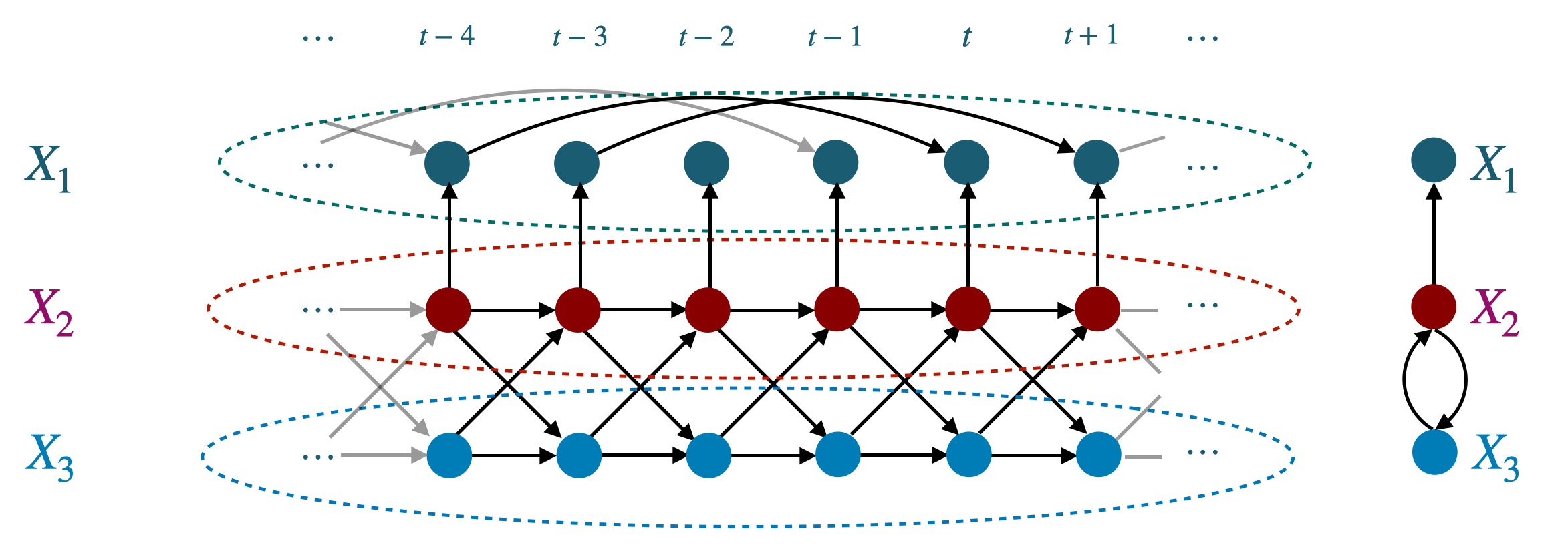}
\caption{The summary graph viewed as a coarsened group DMG of the unrolled time series DMG.} \label{fig.summary-graph-as-grouped-graph}
\end{figure}

Of course, there is no formal reason to disallow more general partitions of $\cV = I \times \mathbb Z$. For instance, when $\mathcal{Q}$ is a partition of the set of processes $\{ X_1,\dots, X_n \} \cong I$, we can define the \emph{grouped ts-DMG} of $\cG$ as $\co(\cG,\mathcal Q')$ where $\mathcal{Q}' = \{ \Y \times \{t \}; \ \Y \in \mathcal{Q}, \ t \in \mathbb{Z} \}$ is the \emph{contemporaneous partition} of $\mathcal{Q}$, see Figure \ref{fig.grouped-ts-DMG}. We can coarsen the grouped ts-DMG further to obtain the \emph{grouped summary DMG} or \emph{grouped process DMG} 
\[ \co(\cG,\mathcal Q')^{\mathrm{sum}} = \co(\cG, \mathcal Q'') \]
 where $\mathcal{Q}'' = \{ \Y \times \mathbb{Z}; \ \Y \in \mathcal{Q}  \} \cong \mathcal{Q}$ is the \emph{full process partition} of $\mathcal{Q}$, see Figure \ref{fig.grouped-summary-DMG}.
 
 \begin{figure}[h!]
 \centering
 \includegraphics[scale=0.35]{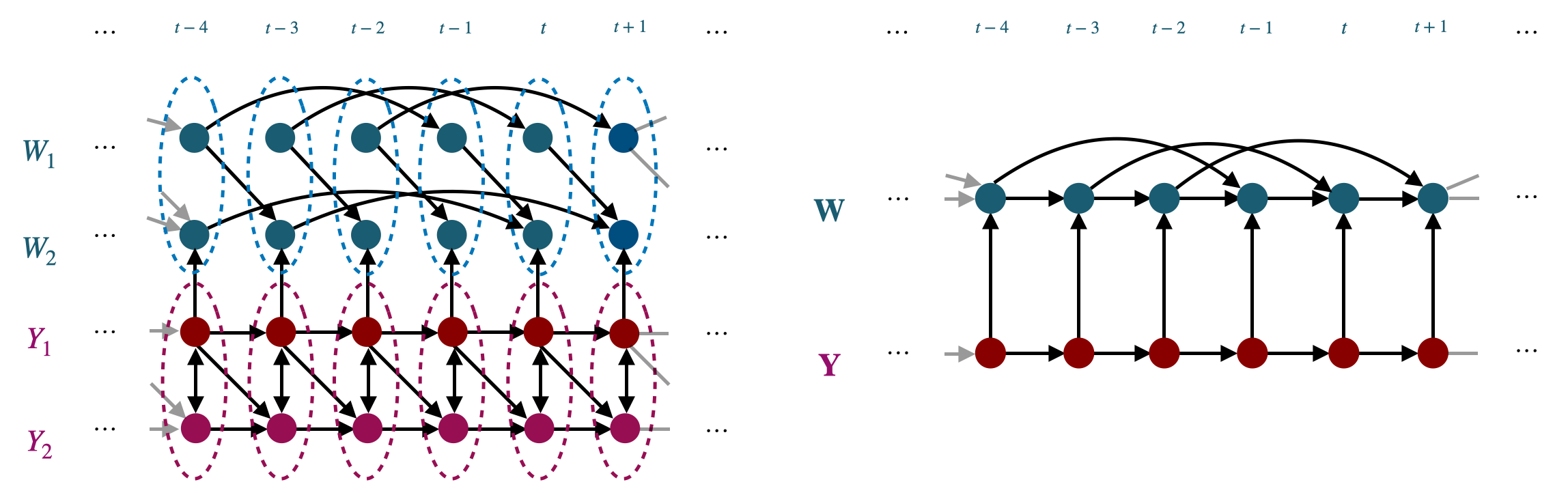}
\caption{Left: A partition $\cP$ of an unrolled time series DMG into contemporaneous groups. Right: The grouped ts-DMG $\co(\cG,\cP)$ with respect to the partition $\cP$.} \label{fig.grouped-ts-DMG}
\end{figure}

 \begin{figure}[h!]
 \centering
 \includegraphics[scale=0.22]{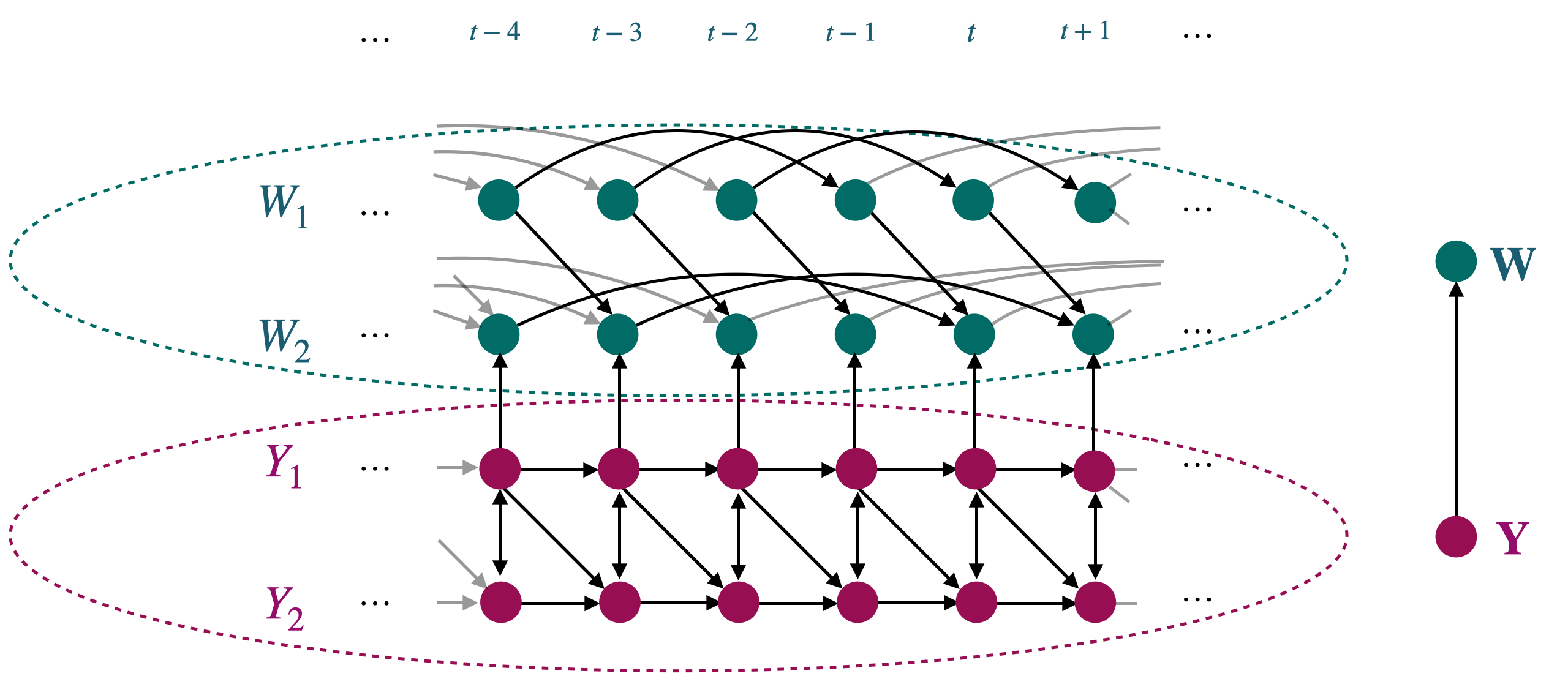}
\caption{Left: A process grouping of a ts DMG. Right: The corresponding grouped summary DMG.} \label{fig.grouped-summary-DMG}
\end{figure}

\subsection{Faithfulness in Grouped Time Series Graphs}

Given that grouped ts-DMGs  and grouped summary DMGs are special cases of coarsened graphs, the criteria of Theorem \ref{lem.connectivity-criterion} and Theorem \ref{lem.connectivity-criterion2} are still sufficient to ensure $\sigma$-faithfulness.

\begin{corollary} \label{cor.faithfulness-ts}
Let $\cG$ be time series DMG and let $\mathcal{Q}$ be a partition of the set of processes $\{ X_1,\dots, X_n \}$ with contemporaneous partition $\mathcal{Q}'$ and full process partition $\mathcal{Q}''$. Moreover, let $P_{\lX}$ be the joint distribution of $\{X_i(t)\}_{i\in I, t \in \mathbb Z}$. If the assumptions of Theorems \ref{lem.connectivity-criterion} or \ref{lem.connectivity-criterion2} are satisfied w.r.t. $\mathcal{Q}'$ ( respectively $\mathcal{Q}''$), then the pair $(\co(\cG, \mathcal Q'), P_{\lX})$ (respectively $(\co(\cG, \mathcal Q''), P_{\lX})$) is $\sigma$-faithful.
\end{corollary}

At the same time, if these criteria are not assumed to hold, violations of $\sigma$-faithfulness are still easily constructed even if there are no contemporeaneous edges and all micro-processes are autocorrelated, see e.g. Figure \ref{fig.ts-faithfulness-violation} for a faithfulness violation on the grouped summary DMG. In addition, in micro-level ts-DMGs, cycles can only appear in the contemporaneous part of the graph as directed edges cannot point backwards in time. Cycles will thus only be included in the grouped time series DMG if the time resolution of the analyzed data is not fine enough to resolve all feedback loops. If the time resolution is believed to be fine enough, all cycles are resolved which renders Theorem \ref{lem.connectivity-criterion} useless in the ts-domain.

 \begin{figure}[h!]
 \centering
 \includegraphics[scale=0.22]{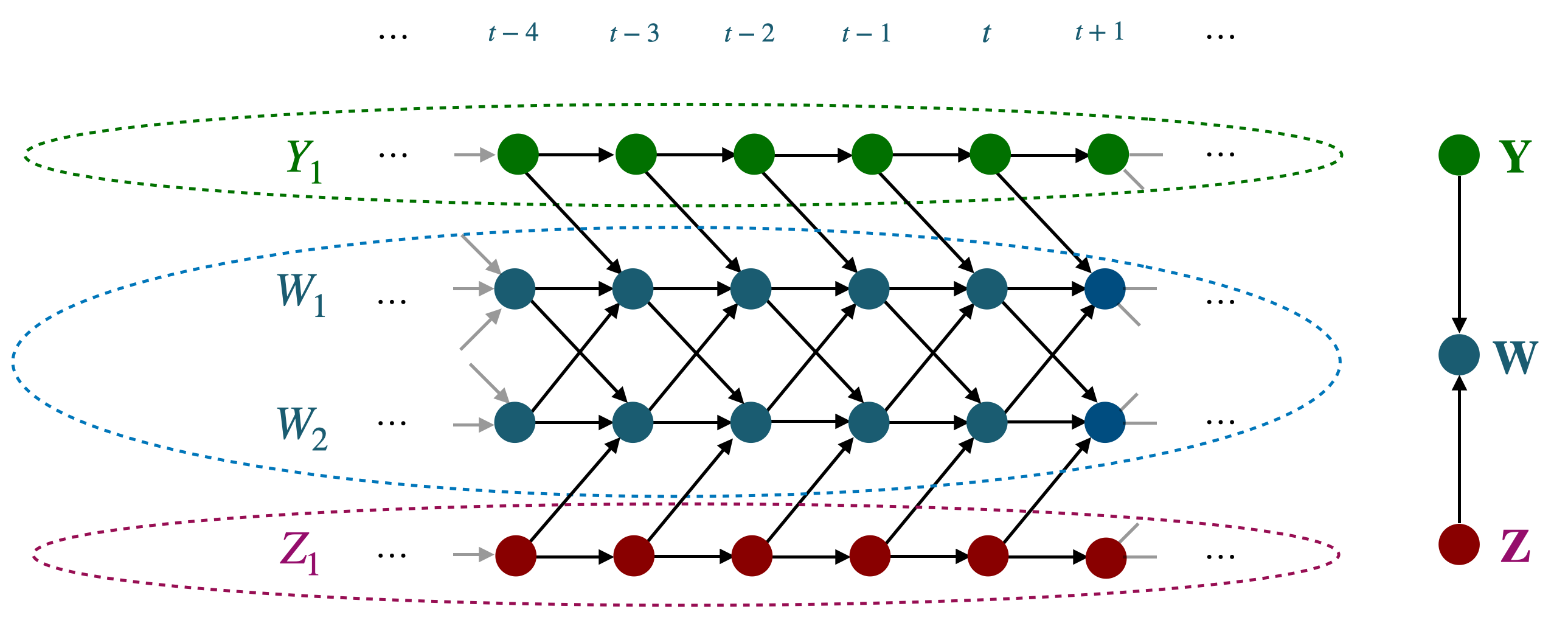}
\caption{A $d$-faithfulness violation on the grouped summary DMG. Conditioning on the process $\W$ blocks all micro paths between the processes $\Y$ and $\Z$.} \label{fig.ts-faithfulness-violation}
\end{figure}

\section{Interpretation of Causation in Group (D)MGs} \label{sec.interpretation}

Many of the examples presented in this work, see e.g. Figures \ref{fig.faithfulness-violations} and \ref{fig.non-local}, show that group DMGs have to be carefully interpreted when associating a causal meaning to paths in the graph; a point that has already been made in \cite{ParKas17}. They formulate a notion of \emph{potential} and \emph{actual causation} in terms of interventions that can be mirrored in our graphical language. 

\begin{definition}[Apparent and true causes] \label{def.actual-potential-causes}
Let $\G$ be a DMG over a set of micro-variables $\lX$ and let $\cP$ be a partition of $\lX$ inducing the group DMG $\co(\cG,\cP)$. We say that $\Y \in \cP$ is an \emph{apparent cause} of $\Z  \in \cP$ if there exists a directed path  $\Y \rightarrow \dots \rightarrow \Z$ on $\co(\cG,\cP)$. $\Y$ is called a \emph{true cause} of $\Z$ if there is a directed path $Y \rightarrow \dots \rightarrow Z$ on $\cG$ for some $Y \in \Y$ and $Z \in \Z$.
\end{definition}

In other words, directed paths on group DMGs may not be regarded as truly causal in general as corresponding micro-paths might be absent. In particular, intervening on a potential cause $\Y$ of $\Z$ might not change the distribution of the effect group $\Z$. We record the following result for good measure.

\begin{lemma} \label{lem.actual-causation}
Let $\G$ be a DMG over a set of micro-variables $\lX$ and let $\cP$ be a partition of $\lX$ inducing the group DMG $\co(\cG,\cP)$. 
\begin{itemize}
\item If $\Y\to \Z$ is a directed edge, then $\Y$ is a true cause of $\Z$.
\item If the condition (ii) and (iii) of Theorem \ref{lem.connectivity-criterion} are satisfied, then any apparent cause of a group $\Z \in \cP$ is a true cause of $\Z$.
\item If the condition (ii) and (iii-a) of Theorem \ref{lem.connectivity-criterion2} are satisfied, then any apparent cause of a group $\Z \in \cP$ is a true cause of $\Z$.
\end{itemize}
\end{lemma}

\begin{proof}
The first claim of the lemma follows directly from the definition of $\co(\cG,\cP)$. The second and third claim follows directly from the proof of Lemmas \ref{lem.connectivity-criterion} and \ref{lem.connectivity-criterion2} where for a given directed path $\Pi = \Y \rightarrow \dots \rightarrow \Z$ on $\co(\cG,\cP)$, we constructed a connecting directed micro-path  $\pi = Y \rightarrow \dots \rightarrow Z$ on $\cG$ for some $Y \in \Y$ and $Z \in \Z$ such that $\co(\pi) = \Pi$.
\end{proof}

\subsection{Causation in Grouped Time Series Graphs}
We now turn to the question whether any apparent cause in a grouped ts-DMGs or a grouped summary DMG is a true cause. For grouped ts-DMGs, the answer is no for the same reason as for usual group DMGs. On the level of the grouped summary graph, however, apparent causation implying true causation may be more realistic, at least if the grouped processes are believed to be \emph{causally mixing}, a notion inspired by the common assumption of mixing in dynamical systems.

\begin{definition} \label{def.causal-mixing}
Consider a ts-DMG $\cG = (\cV,\cE,\cB)$ over micro processes $X_1,\dots, X_n$, $X_i = (X_i(t))_{t\in \mathbb{Z}}$. Let $\mathcal{Q}$ be a partition of $\{ X_1,\dots, X_n \}$ and consider the induced grouped ts-DMG $\co(\cG,\mathcal Q')$ where $\mathcal{Q}' = \{ \Y \times \{t \}; \ \Y \in \mathcal{Q}, \ t \in \mathbb{Z} \}$.
\\
Then, the pair $(\cG,\mathcal{Q})$ is called \emph{causally mixing} if for any $\Y \in \mathcal{Q}$ and any pair of micro-processes $X_i,X_k \in \Y$ the following holds:
\begin{itemize}
\item for any $s\in \mathbb Z$, there exists $t> s$ and a directed path $X_i(s) \to X_{i_1}(s+1) \to X_{i_2}(s+2) \to \dots \to X_{i_m}(t-1) \to X_k(t)$ such that $X_{i_{\alpha}} \in \Y$ for all $\alpha=1,\dots,m$.
\end{itemize}
\end{definition}

Causal mixing means that after a sufficient amount of time has passed, causal information has fully spread throughout any process group. We will see now that causal mixing ensures that, at least at the level of the grouped summary graph, directed causal paths can be understood in the usual sense as any apparent cause is a true cause. However, causal mixing does not ensure $\sigma$-faithfulness on the grouped summary DMG as the example in Figure \ref{fig.ts-faithfulness-violation} demonstrates.

\begin{lemma} \label{lem.grouped-summary-causation}
Consider a stationary ts-DMG $\cG = (\cV,\cE,\cB)$ over micro processes $X_1,\dots, X_n$, $X_i = (X_i(t))_{t\in \mathbb{Z}}$. Let $\mathcal{Q}$ be a partition of $\{ X_1,\dots, X_n \}$ and consider the induced grouped summary DMG $ \tilde{\cG}:= \co(\cG, \mathcal Q'')$ where $\mathcal{Q}'' = \{ \Y \times \mathbb{Z}; \ \Y \in \mathcal{Q}  \}$. If $(\cG,\mathcal{Q})$ is causally mixing, then every apparent cause in $\tilde{\cG}$ is a true cause in $\tilde{\cG}$.
\end{lemma}

\begin{proof}
For this proof, recall that we can identify elements of $\mathcal{Q}$ and $\mathcal{Q}''$ through the map $\Y \mapsto \Y \times \mathbb{Z}$. Consider two process groups $\Y,\Z \in \mathcal{Q}$. Moreover, let $\Pi \times \mathbb{Z} := (\Pi(1) \times \mathbb Z,\mathbf{e}_1,\dots,\mathbf{e}_{r-1},\Pi(r) \times \mathbb Z)$ be a directed path on from $\Pi(1) \times \mathbb{Z} = \Y \times \mathbb{Z}$ to $\Pi(r) \times \mathbb{Z} = \Z \times \mathbb{Z}$ in the group summary DMG $ \tilde{\cG}= \co(\cG, \mathcal Q'')$. We need to show that there exists a micro-path $\pi$ in $\cG$ from $Y(s)$ to $Z(t)$, $s \leq t$, for some micro-processes $Y \in \Y$ and $Z \in \Z$. We construct $\pi$ inductively as follows. First choose a directed micro-edge $e_1 = Y(s) \to W(s_1) \in \mathrm{mic}(\mathbf{e}_1)$ for some $s \leq s_1$. Then, consider $\Pi(i), \ 1 < i < r$ and assume that a directed micro-path $\pi_i$ that ends in $W(s_{i-1}) \in \mathrm{bd}_{\e_{i-1}}(\Pi(i) \times \mathbb{Z})$ has already been constructed. Choose a micro-process $W' \in \Pi(i)$ such that $W'(t') \in \mathrm{bd}_{\e_{i}}(\Pi(i) \times \mathbb{Z})$ for some $t' \in \mathbb Z$. By causal mixing there is a directed path $\xi_i$ from $W(s_{i-1})$ to $W'(t_i)$ for some $t_i > s_{i-1}$ that does not leave $\Pi(i) \times \mathbb{Z}$. Stationarity of $\cG$ and $W'(t') \in \mathrm{bd}_{\e_{i}}(\Pi(i) \times \mathbb{Z})$ imply that also $W'(t_i) \in \mathrm{bd}_{\e_{i}}(\Pi(i) \times \mathbb{Z})$ so we can find a micro-edge $e_{i+1} \in \mic(\mathbf{e}_{i+1})$ whose source node is $W'(t_i)$. After concatenating $ \pi_{i+1} = \pi_i \circ \xi \circ e_{i+1}$ we have obtained the micro-path $\pi_{i+1}$ to $\Pi(i+1) \times \mathbb{\Z}$ and we continue inductively until we reach $\Z \times \mathbb{Z}$.
\end{proof}

\section{Further discussions and Outlook} \label{sec.pitfalls}

In this section, we will zoom out from the technical results of the previous sections and turn towards a high-level discussion on variable groupings and dimension reduction. 

\subsection{Choosing variable groups} 
In this work, we have operated under the standing assumption that the partition $\cP$ of all micro-variables into variable groups is fixed. We have then studied the transferal of causal discovery assumptions from the micro- to the group level given this fixed partition $\cP$. While in many problems, practitioners may have clear ideas on which micro-variables should be grouped together or not, in others there might be more than one plausible choice of partition. When the goal is to make this choice in such a way that faithfulness is a realistic assumption on the group level, Theorems \ref{lem.connectivity-criterion} and \ref{lem.connectivity-criterion2} at least provide a heuristic: there should be sufficient causal interactions internal to the variable groups. In particular, grouping together micro-variables that seem to be unrelated causally, appears to be problematic. This seems to be in line with our intuition. After all, why would one group together variables that seem unrelated in the first place?
Beyond these heuristic considerations, learning pairs $(\cP,\cG(\cP))$ of a partition $\cP$ and a graph $\cG(\cP)$ over its constituents from data under appropriate optimality constraints may be an interesting, although challenging problem for future research.

\subsection{Dimension reduction and causal discovery}

As alluded to in the introduction, in observation-based analyses of causal interactions, the common alternative to working with variable groups in their entirety, is to reduce them to a single univariate variable, or, if they evolve dynamically, to a single index time series. While some form of dimensionality reduction is unavoidable in high-dimensional settings, the goal of this paragraph is to point out the pitfalls of applying a causal discovery method to dimensionally-reduced proxies, at least if dimension reduction is applied naively. In the subsequent paragraph, we will carry out a similar analysis for a second naive approach, namely, using all available micro-variables as the input of a constraint-based causal discovery method. We contrast this to constraint-based group-level discovery, that is the application of a constraint-based method such as the PC algorithm to groups of random variables in which only multivariate conditional independence test between groups are employed as a whole.

\paragraph*{Applying causal discovery to dimensionally reduced variables}

The most common dimension reduction approach to causal discovery on variable groups $\lX_1,\dots,\lX_r$ proceeds as follows. \\

\begin{enumerate}
    \item Reduce $\lX_i$ to a univariate random variable $X_i$, for instance by setting $X_i = m(\lX_i)$ to be the group mean or the first principal component in a PCA on $\lX_i$.
    \item Apply a causal discovery algorithm to $X_1,\dots,X_r$.
\end{enumerate}
\vspace*{0.2cm}
This procedure is appealing to domain researchers for several reasons. First of all, dimension reduction techniques can be carried out quickly, they counter the \emph{curse of dimensionality}, and the resulting quantities can often be interpreted easily. Moreover, as per the law of large numbers, averaging can help to reduce observational noise, at least if noise terms of different members of a given variable group are believed to be weakly correlated. For instance if every member $X^j_i$ of group $\lX_i$ is believed to be produced by a common driver and purely observational noise, i.e. $X^j_i = \hat{X}_i+ \eta_{i,j}$ and the noise terms $\eta_{i,j}$ have mean zero and are weakly  or un-correlated across the $j$ index, then in the large group limit, the group mean $\hat{X}_i = m(\lX_i)$ will recover $\hat{X}_i$. Thus, if the causal dynamics are modelled by structural equations on the $\hat{X}_i$ such as $\hat{X}_i := f_i(\mathrm{pa}(\hat{X}_i), \eta_{\hat{X}_i})$ with $\mathrm{pa}(\hat{X}_i) \subset \{ \hat{X}_1,\dots,\hat{X}_r  \}$ and the groups arise as $\lX_i = (\hat{X}_i,\dots,\hat{X}_i)^T + \boldsymbol{\eta}_{i}$ with zero-mean noise vectors that are mutually independent across the $i$ index and whose components are weakly  or un-correlated, then the group mean will be an appropriate choice of aggregation technique to recover the causal dynamics. 

\noindent On the other hand, if different parts of a given cause group $\Y$ have opposing causal effects on a target group $\Z$ that roughly cancel each other, the effect of the group mean of $\Y$ on the group mean of $\Z$ may be zero, and neither the causal effect nor the dependence $\Y \centernot\ind \Z$ can be recovered from the averaged data.  An often invoked real-world example of this are the opposite-sign effects of two different types of blood cholesterol, low-density lipoprotein (LDL) and high-density lipoprotein (HDL), on heart disease, see \cite{Rubensteinetal17}. Consequently, research on the effect of total blood cholesterol (LDL+HDL) on heart disease has come to contradictory conclusions.

\noindent In a similar vein, conditioning on the mean value $m(\W)$ of a variable group $\W$ may not suffice to recover a conditional independence $\Y \ind \Z | \W$. For instance, consider a structural causal model
\begin{align*}
W_1 \ &= \ \eta_{W_1} \\
W_2 \ &= \ \eta_{W_2} \\
Y \ &:= \ W_1+2W_2+ \eta_{Y} \\
Z \ &:= \ W_1+2W_2+ \eta_{Z},
\end{align*}
with variable partition $\Y = \{Y \}, \Z = \{Z \}, \W = \{W_1,W_2 \}$, and with independent noise terms  $\eta_{W_1},\eta_{W_2}, \eta_{Y},\eta_{Z}$. Then we have $\Y \ind \Z | \W$ but $m(\Y) \centernot\ind m(\Z) \ | \ m(\W)$, where again $m(\cdot)$ denotes the group mean. The latter relation becomes apparent when rewriting $Y = 2m(\W)  + W_2 + \eta_Y$ and $Z = 2m(\W)  + W_2 + \eta_Z$, so that after conditioning on $m(\W)$, $Y$ and $Z$ still share the common random component $W_2$ which is not fully determined by $m(\W)$. Thus, causal discovery approaches that invoke conditional independence tests on aggregated quantities may come to wrong conclusions. However, this example also illustrates that the primary reason for such faulty inferences is that dimension reduction and inference were conducted \emph{independently of each other}. In fact, in the example above, there is an aggregation of $\W$ that does preserve the independence  $\Y \ind \Z | \W$: if $m'(\W) = W_1+2W_2$, then $\Y \ind \Z \ | \ m'(\W)$. Research on how variable aggregation and inference can combined in such a way that they inform each other, is still relatively scarce, and we refer to \cite{ChaEbPer16, ChaEbPer17, Rubensteinetal17} for interesting ideas and further discussions.



\paragraph*{Micro-level causal discovery}
A second straightforward approach to causal discovery on variable groups $\lX_1,\dots,\lX_r$ roughly works as follows: 
\begin{enumerate}
    \item apply a given causal discovery method to the totality of all micro-variables. This will output a graph over all micro-variables containing edges of different types. 
    \item[2a.] Then coarsen this micro-graph as in Definition \ref{def.graphs}, that is draw an edge of a specific type between groups $\Y$ and $\Z$ if there exists an edge of this type between two members $Y \in \Y$ and $Z \in \Z$ of these groups.
    \item[2b.]  Alternatively, if only one edge is to be allowed between groups, decide on the type of this edge by a majority rule, e.g. draw a directed edge $\Y \to \Z$ if the majority of edges between members  $Y \in \Y$ and $Z \in \Z$ are directed as $Y \to Z$.
\end{enumerate} 
As constraint-based causal discovery algorithms such as PC typically come with soundness and completeness guarantees under method-specific assumptions \cite{PearlCausality, KaBue07}, in theory, the micro-graph (and therefore the macro-graph derived from it) can be inferred to an optimal degree, that is up to a certain type of equivalence. Still, in practice, there are some obvious drawbacks of such an approach. First, as the number of micro-variables within groups can be very high, the computational effort can be massive while much of the inferred micro-level information, namely all interaction \emph{internal to variable groups} is of little relevance to the actual task of inferring the interactions \emph{between variable groups}. This issue is particularly problematic if the variable groups happen to be very dense, i.e. if there are many micro-edges within groups. This is because this case falls firmly into the computational worst case scenario for constraint-based causal inference in which computing time grows exponentially with the number of variables \cite{KaBue07}. At the same time, one can argue that typically variable groups are chosen the way they are exactly because their members are highly correlated or have strong causal interactions. From a statistical perspective running many conditional independence tests on the micro-level that are irrelevant to the actual inference task, tends to be detrimental to the method's success, see \cite{wahl_vector_2022} for some toy experiments with two variable groups and continuous data. In addition, the well-known finite sample guarantees of Kalisch and B\"{u}hlmann \cite{KaBue07} for the PC algorithm again rely on sparsity conditions that may not be met on the micro-graph if the variable groups are very dense while they might be met on the coarse group DMG.

On the other hand, full micro-variable causal discovery can sometimes orient edges between groups that a group-level approach can not orient, see Figure \ref{fig.micro-macro-inference}. This can be both a blessing and a curse: while additional orientations are a plus whenever they are correct, a wrong statistical test result of an independence test that only involves micro-variables \emph{within the same group} can lead to a wrongly oriented edge \emph{between} variable groups, see Figure \ref{fig.micro-macro-inference}. Therefore, group-level causal discovery can be considered more conservative than full micro-level causal discovery in the sense that it might provide fewer orientations while being more robust to testing errors. Lastly, if the causal discovery algorithm at hand assumes the absence of hidden confounders, it will suffer if hidden confounding is actually present in the data. Hence, if hidden confounders only affect micro-variables within the same group, then micro-level causal discovery will be challenged while group-level causal discovery will only be affected by confounders between different groups, see again the discussion in Section \ref{sec.Markov}. Nevertheless, in the case of discrete data, conditional independence tests are particularly challenged by large conditioning sets as every state of the conditioning variables has to be considered separately. In this case, the empirical experiments conducted in \cite{ParKas17} suggest that the micro-level causal discovery approach which employs more tests but has smaller conditioning sets than the group-level approach outperforms the latter.

\begin{figure}[h!]
\centering
\includegraphics[scale=0.35]{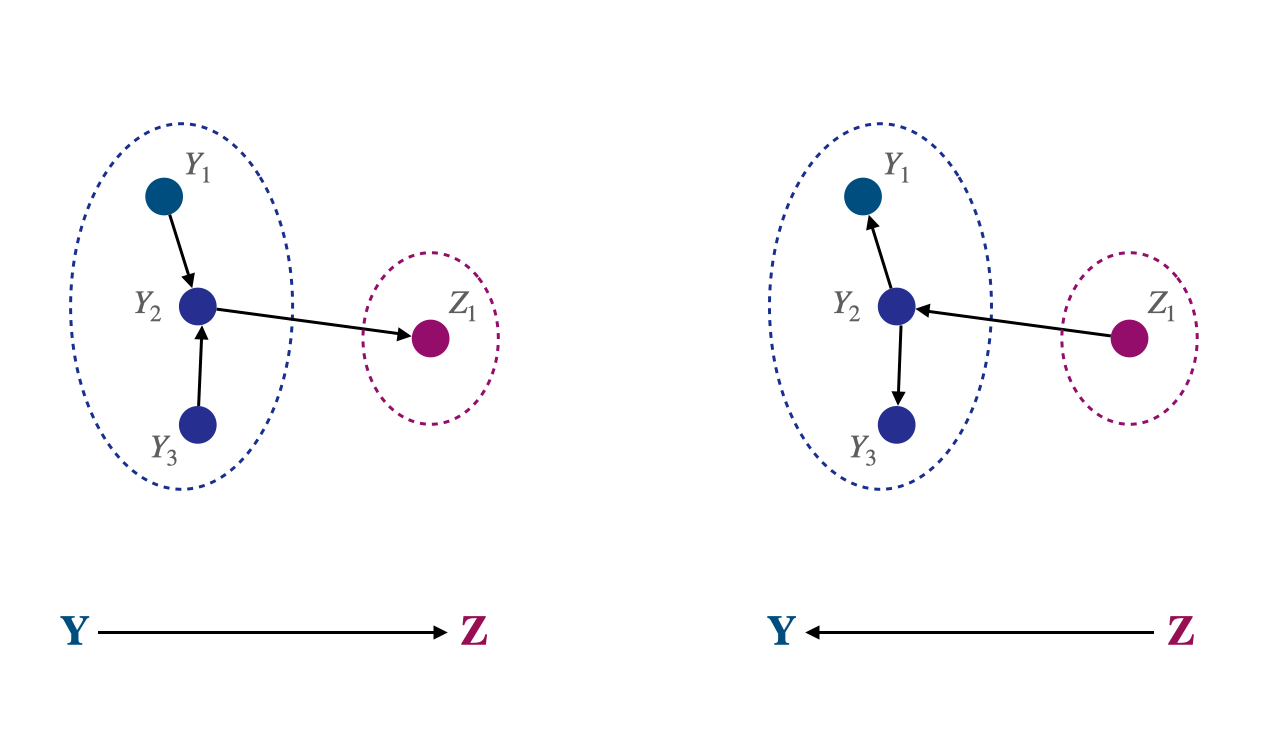}
\caption{Left: Running the PC-algorithm with perfect independence tests on the micro-variables will infer the full micro-structure and will therefore also be able to orient the group-level edge $\Y \to \Z$. Group-level PC will not be able to infer this orientation. Right: If, due to a wrong statistical test result or due to a faithfulness violation, the micro-level PC-algorithm mistakenly judges $Y_1 \ind Y_3$, it has found a separating set for $Y_1$ and $Y_3$ that does not contain $Y_2$ and will thus orient the unshielded triple $Y_1 - Y_2 - Y_3$ as a collider $Y_1 \rightarrow Y_2 \leftarrow Y_3$. If the remaining tests return the true (in)dependecies $Y_1 \ind Z | Y_2, \ Y_3 \ind Z | Y_2$, then PC's orientation rules will imply the edge orientation $Y_2 \to Z$. Hence, the PC-algorithm will again infer the micro-structure on the left and the wrong group-level orientation $\Y \to \Z$.} Note that the wrong test only involves micro-variables that belong to group $\Y$. Group-level PC will never run this wrong test and will not orient the edge  $\Y - \Z$, neither correctly nor wrongly. \label{fig.micro-macro-inference}
\end{figure}

We summarize strengths and pitfalls of dimension reduction causal discovery, micro-level causal discovery as well as group-level causal discovery in Table \ref{tab.strengths}.

\begin{table}[h!]
\centering
    \begin{tabularx}{\textwidth}{lcccrX}
        \toprule
     &\multicolumn{1}{c}{Dimension reduction + CD} & \multicolumn{1}{c}{Micro-level CD} & \multicolumn{1}{c}{Group-level CD} \\
        \midrule
      \textbf{Strengths} & Computationally most efficient & Good for small groups; & Fewer CI tests than micro-level CD;\\
        &approach; & empirically superior  & robust to within-group confounding  \\
      &noise-removal. & to group-level CD on discrete data. &  and other violations. \\
      \midrule
        \textbf{Weaknesses}  & May change conditional & Computationally inefficient;& Assumptions and interpretation of \\ 
        & independencies and & vulnerable to & output must be evaluated carefully; \\
        &  causal conclusions & within-group assumption & multivariate CI testing  \\
        & fundamentally. & violations. &  less developed; \\
        & \ & \ & computationally less efficient \\
        & \ & \ & than dimension reduction + CD. \\
        \bottomrule
    \end{tabularx}
    \caption{Strengths and weaknesses of the three fundamental approaches to causal discovery for variable groups: causal discovery after dimension reduction, micro-level causal discovery, and group-level causal discovery. Approaches that integrate dimension reduction and inference, while perhaps retaining reduced variable groups of smaller size might be a fruitful middle ground.}
    \label{tab.strengths}
\end{table}


\section{Summary}

In this work, we have provided a thorough discussion of assumptions for causal discovery on groups of random variables. In particular we have shown that causal faithfulness is easily violated in generic settings so that faithfulness-based causal discovery methods need to be applied with care. On the other hand we have presented two criteria (Theorem \ref{lem.connectivity-criterion} and \ref{lem.connectivity-criterion2}) on the internal connectivity of variable groups that do guarantee $\sigma$-faithfulness. It will be important to develop and evaluate more elaborate group-level causal discovery techniques and to compare them to the baseline methods presented in Section \ref{sec.pitfalls} empirically, in particular for continuous data. On the theoretical side, it would be worthwhile to study the compatibility of statistical dimension reduction and causal modelling in greater detail, for instance following the ideas laid out in \cite{ChaEbPer16, ChaEbPer17, Rubensteinetal17}.




\paragraph*{Acknowledgements} 
J.W., U.N., and J.R. received funding from the European Research Council (ERC) Starting Grant CausalEarth under the European Union’s Horizon 2020 research and innovation program (Grant Agreement No.\ 948112). The authors thank Sofia Faltenbacher for designing the layout of many of the figures in this work.


\bibliographystyle{vancouver}
\bibliography{References.bib}

\newpage

\begin{appendices}
\section{Group DMGs from group-valued SCMs} \label{sec.group_SCM}

In this appendix, we will shortly discuss another way of obtaining a group DMG that is distinct from coarsening a graph of micro-variables, namely by defining a model directly through structural equations. For a discussion of counterfactual distributions in vector-valued SCMs, see \cite[Supplement, Theorem 7]{anand_causal_2023}.

\begin{definition}[vector-valued SCMs]
A \emph{vector-valued structural causal model (vSCM)} $\mathfrak{M} = (\mathfrak{S},P_{\mathbf{E}})$ over a partition $ \cP $ of a set of random variables $\lX$ into random vectors $\lX^1,\dots,\lX^r$ is a collection of structural assigments
\[
  \lX_i :=  \mathbf{f}^i(\mathrm{pa}(\lX^i),\mathbf{E}^i) 
\]
with $\mathrm{pa}(\lX^i) \subset \{\lX^1,\dots,\lX^r\} \backslash \{ \lX^i \}$ and multivariate noise vectors $\mathbf{E}^1,\dots, \mathbf{E}^r$ with $\dim(\lX^i) = \dim(\mathbf{E}^i)$ that have joint distribution $P_{\mathbf{E}}$.
The causal graph $\cG(\mathfrak{M})$ of $\mathfrak{M}$ is the DMG with nodes $\lX^1,\dots,\lX^r$ where a directed edge $\lX^j \to \lX^i$ is drawn if $\lX^j \in \mathrm{pa}(\lX^i)$ and a bidirected edge $\lX^i \leftrightarrow\lX^j$ is drawn if $\mathbf E^i \centernot \ind \mathbf E^j$.\footnote{If the causal graph $\cG(\mathfrak{M})$  has cycles, it is not always true that random vectors $\lX^1,\dots,\lX^r,\mathbf{E}^1,\dots, \mathbf{E}^r$ obeying the SCM actually exist. Solvability of cyclic SCMs is thoroughly discussed in \cite{BonFOrPetMoo21} for univariate node variables and the results straightforwardly transfer to the multivariate setting. The same is true for their analysis on Markov properties if mutual independence is used as an independence model.}
\end{definition}

While group DMGs derived by coarsening micro-variable graphs assume a causal structure on the level of the micro-variables and is then "forgotten" after coarsening, in a vector-valued SCM any causal meaning in the form of a graph is only defined on the group-level. The internal relationships among the entries of a vector $\lX^i$ that are not due to external influences are modelled only by the distribution $P_{\mathbf{E}^i}$ and are thus of a probabilistic nature. This seems reasonable for many practical applications where the micro-variables may not be considered causal entities (for instance imagine $\lX^i$ to be a field of surface tempature measurements in some spatial region). On the other hand, vector-valued SCMs make it hard to \emph{derive} faithfulness results from properties of the micro-variables, as no notion of faithfulness is purely distributional. Instead faithfulness can only be postulated as an assumption on the group-level directly.

\section{Proofs}

\subsection*{Proofs of the results in Section \ref{sec.groupDMGs}} \label{app.proofs_sec_groupDMGs}

\begin{proof}[Proof of Lemma \ref{lem.acyclicparts}] \ 
\begin{itemize}
\item[(i)] Assume first that $\mathrm{co}(\G,\cP)$ is acyclic and let $\cW$ be a strongly connected component. If there were $W_1, W_2 \in \cW $ that belonged to different groups of the partition $\cP$, say $W_1 \in \Y$ and $W_2 \in \Z$, then on $\G$ we could find directed paths $\pi_1$ from $W_1$ to $W_2$ and $\pi_2$ from $W_2$ to $W_1$. Then the induced coarse path $\co(\pi_1)$ would constitute a directed path from $\Y$ to $\Z$ and the induced coarse path $\co(\pi_2)$ would constitute a directed path from $\Z$ to $\Y$. Concatening both paths, we would obtain a cycle which contradicts our assumption. 
\item[(ii)] The converse is already wrong for coarsenings of micro DAGs in which the strongly connected components correspond to the nodes of the graph, see e.g. Figure \ref{fig.cycles}.
\item[(iii)] Let $\cP$ be the partition of $\G$ into strongly connected component and let $\tilde{\pi}$ be a directed path from $\Y$ to $\Z$. Then, we argue first that for any two node $Y \in \Y, \ Z \in \Z$, there is a directed micro path $\pi$ from $Y$ to $Z$ on $\cG$. Indeed, if $\tilde{\pi}$ just consists of an edge $\Y \to \Z$, then there must be $Y' \in \Y$ and $Z' \in \Z$ that are connected by a micro edge $Y' \to Z'$. By the definition of strongly connected components, there must also be directed paths from $Y$ to $Y'$ and from $Z'$ to $Z$, so we have found the desired micro path. If $\tilde{\pi}$ has more than one edge, we can proceed similarly by noting that for any motive $\W \rightarrow \Y \rightarrow \Z$ there are micro edges $W \to Y$, $Y' \to Z$ with $W \in \W, \ Y,Y' \in \Y, \ Z \in \Z$ and either $Y=Y'$ or there is a directed path from $Y$ to $Y'$ as $\Y$ is strongly connected. Concatenating all edges and paths found this way, we obtain the desired micro path. 
Finally, we conclude by observing that any cycle on $\co(\G,\cP)$ must thus induce a cycle on the micro MG $\cG$. Indeed, a cycle on $\co(\G,\cP)$ could be decomposed into directed paths $\tilde{\pi}_1$ and $\tilde{\pi}_2$ one from say $\Y$ to $\Z$ and one from $\Z$ to $\Y$ to which we then apply the argument above.
\end{itemize}
\end{proof}

\begin{proof}[Proof of Theorem \ref{thm.coarsen_acyclify}]
Write $\cG = (\cV,\cE,\cB, \cU)$ and  $\cG^{\mathrm{acy}} = (\hat{\cV},\hat{\cE},\hat{\cB},\hat{\cU})$ and $\cP = \{\lX^1,\dots,\lX^r \}$. We have to show that $\mathrm{co}(\cG^{\mathrm{acy}},\cP)$ and $\mathrm{co}(\cG,\cP)$ have the same directed, bidirected and undirected edges. \\
First, let $\lX^i \to \lX^j$ be a directed edge on $\mathrm{co}(\cG,\cP)$ so that there must exist a directed edge $A \to B \in \cE$ with $A \in \lX^i$ and $B \in \lX^j$. Hence $A \in \mathrm{pa}_{\G}(B) \subset \mathrm{pa}_{\G}(\mathrm{sc}_{\G}(B))$ and we also see that $A \notin \mathrm{sc}_{\G}(B)$ by part (a) of Lemma \ref{lem.acyclicparts} as $\cP$ was assumed acyclic w.r.t. $\cG$. So by definition of acyclification, we get $A \to B \in \hat{\cE}$ and thus the edge $\lX^i \to \lX^j$ is present on $\mathrm{co}(\cG^{\mathrm{acy}},\cP)$. 
On the other hand, if  $\lX^i \to \lX^j$ is a directed edge on $\mathrm{co}(\cG^{\mathrm{acy}},\cP)$, then there must be an edge $A \to B \in \hat{\cE}$ with $A \in \lX^i$ and $B \in \lX^j$. Therefore $A \in \mathrm{pa}_{\G}(\mathrm{sc}_{\G}(B))\backslash \mathrm{sc}_{\G}(B) $, so there must be a node $C \in \mathrm{sc}_{\G}(B)$ and an edge $A \to C \in \cE$. By Lemma \ref{lem.acyclicparts} (a), we obtain $\mathrm{sc}_{\G}(B) \subset \lX^j$ so that there must be an edge $\lX^i \to \lX^j$ on $\mathrm{co}(\cG,\cP)$. \\
We now turn to bidirected edges. If $\lX^i \leftrightarrow \lX^j$ is a bidirected edge on $\mathrm{co}(\cG,\cP)$, then there exists a bidirected edge $A \leftrightarrow B \in \cB$ with $A \in \lX^i$ and $B \in \lX^j$. By definition of acyclification, we also have $A \leftrightarrow B \in \hat{\cB}$, so $\lX^i \leftrightarrow \lX^j$ is a bidirected edge on $\mathrm{co}(\cG^{\mathrm{acy}},\cP)$ as well. Finally assume that $\lX^i \leftrightarrow \lX^j$ is a bidirected edge on $\mathrm{co}(\cG^{\mathrm{acy}},\cP)$, so that there exists a bidirected edge $A \leftrightarrow B \in \hat{\cB}$ with $A \in \lX^i$ and $B \in \lX^j$. By acyclicity and Lemma \ref{lem.acyclicparts} $A$ and $B$ must lie in different strongly connected components of $\cG$. Therefore there must be $A'\in \mathrm{sc}_{\cG}(A) \subset \lX^i$ and $B'\in \mathrm{sc}_{\cG}(B) \subset \lX^j$ connected by a bidirected edge $A' \leftrightarrow B' \in \cB$. We conclude that $\lX^i \leftrightarrow \lX^j$ must be a bidirected edge on $\mathrm{co}(\cG,\cP)$.
Finally, we discuss undirected edges. Thus assume first that $\lX^i - \lX^j$ is an undirected edge on $\mathrm{co}(\cG,\cP)$, so that there must exist an undirected edge $A - B \in \cU$ with $A \in \lX^i$ and $B \in \lX^j$. Since $\cP$ was assumed acyclic,  we see that $A \notin \mathrm{sc}_{\G}(B)$ by part (a) of Lemma \ref{lem.acyclicparts}, so that there must be an undirected edge $A - B \in \hat{\cU}$. Thus $\lX^i - \lX^j$ must be an undirected edge of $\mathrm{co}(\cG^{\mathrm{acy}},\cP)$. Conversely if $\lX^i - \lX^j$ is assumed to be an undirected edge of $\mathrm{co}(\cG^{\mathrm{acy}},\cP)$, there must be an undirected edge $A - B \in \hat{\cU}$ with $A \in \lX^i, \ B \in \lX^j$. By definition of acyclification, we must have $A-B \in \cU$ and thus $\lX^i - \lX^j$ must be an undirected edge of $\mathrm{co}(\cG,\cP)$. This finishes the proof.

\end{proof}

\begin{proof}[Proof of Lemma \ref{lem.d-sep}]
We will only discuss the case where $\co(\pi)$ (and thus $\pi$) is a non-trivial walk.
\begin{itemize}
\item[(i)] If $\co(\pi)$ is $\sigma$-blocked by $\cS$, then there are three options.
\begin{itemize}
\item[(1)]   If the first (or last) node of $\co(\pi)$ is in $\cS$, then $\cT$ must contain the first (or last) node of $\pi$ and thus $\sigma$-blocks $\pi$.
\item[(2)]  There is a collider $\W$ on $\co(\pi)$ with $\cS \cap \des(\W) = \emptyset$. We argue that in this case $\W$ must contain a collider $W$ of the micro walk $\pi$. Indeed, if $\pi$ passes through only one node of $\W$, this follows directly. If $\pi$ passes through more than one node, $\pi$ must enter $\W$ at a micro node $\pi(i)$ with an edge pointing to $\pi(i)$ (either bidirected or directed) and leave $\W$ at a micro node $\pi(j)$, $j > i$, again with an edge pointing to $\pi(j)$ (either bidirected or directed). Thus at some point in the path segment $\pi(i,j)$ the directionality of the arrows must oppose each other, that is to say that path segment must contain a collider, say $\pi(l)$. Any descendant $D$ of $\pi(l)$ must lie in $\W$ itself or in a proper descendant of $\W$, say $D \in \mathbf{D}$ as the directed path $\pi(l) \to \dots \to D$ induces a coarse path $\W \to \dots \to \mathbf{D}$. As both $\W$ and its proper descendants do not lie in $\cS$, $\des(\pi(l)) \cap \cT = \emptyset$.
\item[(3)] There is a non-collider $\co(\pi)(k)$ on $\co(\pi)$ that is contained in $\cS$ and an edge $\co(\pi)(k) \to \co(\pi)(l), l \in \{k-1,k+1 \}$ with $\scc_{\co(\G,\cP)}(\co(\pi)(k)) \neq \scc_{\co(\G,\cP)}(\co(\pi)(l)) $. Therefore, on $\pi$, there must be an edge $\pi(i) \to \pi(j), \ j \in \{i-1,i+1 \}$ with $\pi(i) \in \co(\pi)(k)$ and $\pi(j) \in \co(\pi)(l)$. Since $\pi(i)$ has an outgoing edge it is a non-collider and by the contraposition of Lemma \ref{lem.sc-inherit} $\scc_{\G}(\pi(i)) \neq \scc_{\G}(\pi(j))$. Therefore $\cT$, which contains $\pi(i)$ $\sigma$-blocks $\pi$.

\end{itemize}
\item[(ii)] We use the following counterexample to show that the converse of (i) is not true. Let $\G$ be given by $W \rightarrow Y_1 \rightarrow Y_2 \leftarrow Y_3 \rightarrow Z$ partitioned as $\W = \{W\}, \Y = \{ Y_1, Y_2, Y_3 \}, \Z = \{ Z \} $. Then the path from $W$ to $Z$ is closed since it contains the collider $Y_2$ while the coarse path $\W \to \Y \to \Z$ is open.
\item[(iii)] If $\pi$ is an arbitrary walk between $Y \in \Y$ and $Z \in \Z$, then $\co(\pi)$ is a walk between $\Y$ and $\Z$ and thus $\sigma$-blocked by $\cS$. Hence by assertion (i), $\pi$ is $\sigma$-blocked by $\cT$.
\item[(iv)]  Let $\G$ be given by $W \rightarrow Y_1 - Y_2 - Y_3 \leftarrow Z$ partitioned as $\W = \{W\}, \Y = \{ Y_1, Y_2, Y_3 \}, \Z = \{ Z \} $. Then the path $\W \to \Y \leftarrow \Z$ is $\sigma$-blocked while the micro path from $W$ to $Z$ is $\sigma$-open as it does not contain any colliders.
\end{itemize}
\end{proof}

\begin{proof}[Proof of Lemma \ref{cor.m-sep}]
For part (i) it suffices to note that the only difference between the two types of separation lies in the their definition for non-colliders, so only part (i)(3) of the proof of Lemma \ref{lem.d-sep} slightly differs. When a coarse walk $\co(\pi)$ has a non-collider, say $\co(\pi)(k)$, then $\pi$ must have a non-collider $\pi(j) \in \co(\pi)(k)$ as well. So if $\co(\pi)(k) \in \cS$, then $\pi(j) \in \cT$ and $\pi$ is $m$-blocked. Part (iii) follows directly from (i) and the counterexamples of parts (ii) and (iv) do not involve cycles and are equally valid for $m$-separation.
\end{proof}

\subsection{Proofs of the results in Section \ref{sec.Markov}} \label{app.proofs_sec_Markov}

\begin{proof}[Proof of Lemma \ref{lem.mutual-indep-identity2}]
We prove Lemma \ref{lem.mutual-indep-identity2} by induction over $n = |\Y \cup \Z|$. For $n=2$, we must have $\Y = \{ Y \}, \Z = \{ Z \}$ and the result follows by choosing the subset $\cM = \emptyset$. Now assume that the result has been shown for some arbitrary but fixed $n\geq 2$ and let $|\Y \cup \Z| = n+1$. W.l.o.g. we can assume that $|\Z| > 1$. Let $Y \in \Y$ and $Z \in \Z$ be arbitrary. Choosing $\cM = \Z\backslash\{ Z \}$, by assumption we have $Y \ind Z | \W, \Z\backslash\{ Z \}$. According to Lemma \ref{lem.mutual-indep-identity}, we are done if we can show that also $Y \ind \Z\backslash\{ Z \} | \W$. To prove this, we observe first that $\Y \cup \Z\backslash\{ Z \} = n$. Moreover, for any $Y \in \Y$, $Z' \in \Z\backslash\{ Z \}$ and $\mathcal{M}' \subset \Y \cup \Z\backslash\{ Z \} \backslash\{Y,Z' \} $, we have $Y \ind Z' \ | \ \W, \mathcal{M}'$. Thus the induction hypothesis implies that $\Y \ind  \Z \backslash\{ Z \} | \W$ and in particular $Y \ind  \Z \backslash\{ Z \} | \W$ as desired.
\end{proof}

\begin{proof}[Proof of Lemma \ref{lem.equivalent-independence}]
We only have to prove that conditional pairwise independence implies conditional mutual independence. So let $\Y, \Z$ be pairwisely independent given $\cT$. Iterating part (ii) of Lemma \ref{lem.mutual-indep-identity}, mutual independence follows if we can show that for any pair $Y \in \Y, Z \in \Z$ and any subset $\mathcal{M} \subset \Y \cup \Z$, we also have $Y \ind \Z | \cM \cup \cT$.
By $\sigma$-faithfulness of $(\G,P_{\lX})$, any pair $Y \in \Y, Z \in \Z$ is $\sigma$-separated by $\cT$, that is all micro-paths leading from the group $\Y$ to the group $\Z$ are $\sigma$-blocked by $\cT$. If we can show that all micro-paths are still $\sigma$-blocked by $\cM \cup \cT$, the result follows by the $\sigma$-Markov property of $(\G,P_{\lX})$. Since all such micro-paths are $\sigma$-blocked by $\cT$, we only need to make sure that none of these paths is opened again by adding $\cM$ to the separating set. Suppose there was such a path $\pi$ starting at $\pi(1) \in \Y$ and ending at $\pi(r) \in \Z$, that is $\sigma$-blocked by $\cT$ but $\sigma$-unblocked by $\cT \cup \cM$. Let $\pi(k)$ be the last node of $\pi$ in $\Y$ and let $\pi(l), k < l$ be the first node of $\pi$ in $\Z$ after $k$. Since $\cT \subset \lX \backslash \{\Y, \Z \}$, the subpath $\pi'$ of $\pi$ starting at $\pi'(1) = \pi(k)$ and ending at $\pi'(s) = \pi(l)$ must still be $\sigma$-blocked by $\cT$. On the other hand, since $\cM\cup \cT$ $\sigma$-unblocks $\pi$, it must $\sigma$-unblock $\pi'$. Thus there must be at least one collider on $\pi'$ that has a descendant in $\cM$. Let $\pi'(i)$ be the last such collider on $\pi'$ with descendant $D \in \cM$. If $\pi'(i) = D$, then $\pi'(i) \in \Y \cup \Z$ contradicting the fact that $\pi'$ does not have inner nodes in $\Y \cup \Z$. Therefore $D$ must be a proper descendant of $\pi'(i)$. As $D \in \cM$, we have in particular $D \in \Y \cup \Z$ and we will assume w.l.o.g. $D \in \Y$. Let $\pi'' = \pi(i) \to \dots \to D$ be the descending micro-path and assume that $\pi''(j)$ is the first node that belongs to $\cM$. But then, the concatenation of $\pi''(j) \leftarrow \dots \leftarrow \pi'(i)$ and $\pi'(i,s)$ leads from $\Y$ to $\Z$ and is $\sigma$-unblocked by $\cT$ contradicting our assumption that all such micro-paths must be $\sigma$-blocked by $\cT$.
\end{proof}

\begin{proof}[Proof of Theorem \ref{prop.sigma-Markov_relations}]
We assume that $P_{\lX}$ does not have the strong $\sigma$-Markov property and show that this leads to a contradiction. Since $(\co(\G,\cP),P_{\lX})$ does not have the $\sigma$-Markov property, there must be vectors $\Y, \Z$ that are $\sigma$-separated by a set of vectors $\cS$ but not mutually conditionally independent given $\cS$.
By Lemma \ref{lem.d-sep}, for every pair $Y \in \Y, Z \in \Z$, all paths on $\G$ between $Y$ and $Z$ are $\sigma$-blocked by $\mathcal{T}$. On the other hand, using Lemma  \ref{lem.mutual-indep-identity2}, we see that there must be $Y' \in \Y, Z' \in \Z $ and a subset $\mathcal{M} \subset \Y \cup \Z$ such that 
\[Y' \centernot\ind Z' | \mathcal{M}, \mathcal{T}. \]
By the $\sigma$-Markov property on the micro DMG $\G$ this means that there must be a path $\pi$ between $Y'$ and $Z'$ on $\cG$ that is not $\sigma$-blocked by $\mathcal{M},\mathcal{T}$, but is $\sigma$-blocked by $\mathcal{T}$. Therefore $\Y \notin \cS$ and $\Z \notin \cS$ as otherwise we would have $Y' \in \cT$ or $Z' \in \cT$. Moreover $\cT$ cannot contain any non-collider $\pi(l)$ of $\pi$ pointing to a neighbor $\pi(l \pm 1)$ in a different strongly connected component, $\pi$ must have at least one collider and any collider on $\pi$ must have a descendant in $\mathcal M$. Our goal is now to construct a path $\tilde{\pi}$ on $\co(\G,\cP)$ that is not $\sigma$-blocked by $\cS$ resulting in a contradiction. Consider first the coarsened path $\co(\pi)$ of $\pi$. $\cS$ can not contain any non-colliders of $\co(\pi)$ pointing to a neighbor in a different strongly connected component of $\co(\G,\cP)$ by Lemma \ref{lem.sc-inherit}. If $\co(\pi)$ does not have any colliders it must be $\sigma$-unblocked by $\cS$ and we are done. Therefore, assume that $\co(\pi)$ does contain colliders. If all such colliders would have a descendant in $\cS$, again the path would be $\sigma$-unblocked by $\cS$ and the desired contradiction would be obtained. Thus, we can assume that at least one collider on $\co(\pi)$ does not have any descendants in $\cS$. Any collider $\mathbf{C}$ of $\co(\pi)$ must contain a micro-collider $C$ of the micro-path $\pi$ and by the considerations above there must be a directed micro-path $C \to \dots \to M$ for some $M \in \cM \subset \Y \cup \Z$.  Coarsening this micro-path we see that for any collider $\mathbf{C}$ of $\co(\pi)$ there must be a directed macro-path from $\mathbf{C}$ to $\Y$ or to $\Z$. Writing out $\co(\pi)= (\co(\pi)(1),\dots,\co(\pi)(r))$ where $\co(\pi)(1) = \Y$ and $\co(\pi)(r) = \Z$, we define the sets
\[
 U = \left\{ k \ | \ \co(\pi)(k) \in \mathrm{col}(\co(\pi)), \ \cS \cap  \mathrm{des}(\co(\pi)(k)) = \emptyset , \  \mathrm{and} \ \Z \in \mathrm{des}(\co(\pi)(k))  \right\}
\]
and
\[
 U' = \left\{ k \ | \ \co(\pi)(k) \in \mathrm{col}(\co(\pi)), \ \mathrm{des}(\co(\pi)(k)) = \emptyset ,\mathrm{and} \ \Y \in \mathrm{des}(\co(\pi)(k))  \right\}.
\]
By the considerations above, at least one of these sets must be non-empty. If $U$ is non-empty, let $k'$ be its minimum, so that $\co(\pi)(k)$ is the collider closest to $\Y$. Thus the subpath path $\co(\pi)(1),\dots,\co(\pi)(k'))$ must be right- directed,i.e. $\co(\pi)(1) \to \dots \to \co(\pi)(k')$. Since $k' \in U$, we can join it with a directed path from $\co(\pi)(k')$ to $\Z$ yielding a path $\Y \to \dots \to \Z$. This path can not be $\sigma$-blocked by $\cS$ as all of its nodes are either non-colliders of $\co(\pi)$ or descendants of $\co(\pi)(k')$. Thus we have found the desired path. If $U'$ is non-empty the argument is analoguous with $k'':= \max U'$ instead of $k'$ and left-directed instead of right-directed paths.
\end{proof}

\begin{proof}[Proof of Theorem \ref{prop.m-Markov_relations}]
Up to a few subtleties, the proof is similar to the one of Theorem \ref{prop.sigma-Markov_relations}. 
We assume that $P_{\lX}$ that $(\co(\G,\cP),P_{\lX})$ does not have the $m$-Markov property. We show that this leads to a contradiction. Since $(\co(\G,\cP),P_{\lX})$ does not have the $m$-Markov property, there must be vectors $\Y, \Z$ that are $m$-separated by a set of vectors $\cS$ but not mutually conditionally independent given $\cS$. 
 Hence, by Lemma \ref{cor.m-sep}, for every pair $Y \in \Y, Z \in \Z$, all paths on $\G$ between $Y$ and $Z$ are $m$-blocked by $\mathcal{T}$. On the other hand, using Lemma \ref{lem.mutual-indep-identity2}, we see that there must be $Y' \in \Y, Z' \in \Z $ and a subset $\mathcal{M} \subset \Y \cup \Z$ such that 
\[Y' \centernot\ind Z' | \mathcal{M}, \mathcal{T}. \]
By the $m$-Markov property on $\G$ this means that there must be a path $\pi$ between $Y'$ and $Z'$ on $\cG$ that is not $m$-blocked by $\mathcal{M},\mathcal{T}$, but is $m$-blocked by $\mathcal{T}$. Therefore $\cT$ cannot contain any non-colliders of $\pi$, $\pi$ must have at least one collider and any collider on $\pi$ must have a descendant in $\mathcal M$. Our goal is now to construct a path $\tilde{\pi}$ on $\co(\G,\cP)$ that is not $m$-blocked by $\cS$ resulting in a contradiction. Consider first the coarsened path $\co(\pi)$ of $\pi$. $\cS$ can not contain any non-colliders of $\co(\pi)$ as otherwise $\cT$ would contain a non-collider of $\pi$. So if $\co(\pi)$ does not have any colliders it must be $m$-open given $\cS$ and we are done. So assume that $\co(\pi)$ does contain colliders. If all such colliders would have a descendant in $\cS$, again the path would be $m$-opened by $\cS$ and the desired contradiction would be obtained. Thus assume that at least one collider on $\co(\pi)$ does not have any descendants in $\cS$. Any collider $\mathbf{C}$ of $\co(\pi)$ must contain a micro-collider $C$ of the micro-path $\pi$ and by the considerations above there must be a directed micro-path $C \to \dots \to M$ for some $M \in \cM \subset \Y \cup \Z$.  Coarsening this micro-path we see that for any collider $\mathbf{C}$ of $\co(\pi)$ there must be a directed macro-path from $\mathbf{C}$ to $\Y$ or to $\Z$. Writing out $\co(\pi)= (\co(\pi)(1),\dots,\co(\pi)(r))$ where $\co(\pi)(1) = \Y$ and $\co(\pi)(r) = \Z$, we define the sets
\[
 U = \left\{ k \ | \ \co(\pi)(k) \in \mathrm{col}(\co(\pi)), \ \cS \cap  \mathrm{des}(\co(\pi)(k)) = \emptyset , \  \mathrm{and} \ \Z \in \mathrm{des}(\co(\pi)(k))  \right\}
\]
and
\[
 U' = \left\{ k \ | \ \co(\pi)(k) \in \mathrm{col}(\co(\pi)), \ \mathrm{des}(\co(\pi)(k)) = \emptyset ,\mathrm{and} \ \Y \in \mathrm{des}(\co(\pi)(k))  \right\}.
\]
By the considerations above, at least one of these sets must be non-empty. If $U$ is non-empty, let $k'$ be its minimum, so that $\co(\pi)(k)$ is the collider closest to $\Y$. Thus the subpath path $\co(\pi)(1),\dots,\co(\pi)(k'))$ must be right- directed,i.e. $\co(\pi)(1) \to \dots \to \co(\pi)(k')$. Since $k' \in U$, we can join it with a directed path from $\co(\pi)(k')$ to $\Z$ yielding a path $\Y \to \dots \to \Z$. This path can not be $m$-blocked by $\cS$ as all of its nodes are either non-colliders of $\co(\pi)$ or descendants of $\co(\pi)(k')$. Thus we have found the desired path. If $U'$ is non-empty the argument is analoguous with $k'':= \max U'$ instead of $k'$ and left-directed instead of right-directed paths.
\end{proof}

\subsection{Proofs of the results in Section \ref{sec.faithfulness}} \label{app.proofs_sec_faithfulness}

\begin{proof}[Proof of Lemma \ref{lem.macro-micro-faithful}]
Assume first that $(\co(\G,\cP),P_{\lX})$ is $\sigma$-faithful and let be $\Y$ and $\Z$ $\sigma$-connected by a set $\cS \subset \cP$. Therefore by assumption $\Y \centernot\ind \Z | \cT$. By Lemma \ref{lem.equivalent-independence} and our assumptions, conditional pairwise independence of groups implies conditional mutual independence. Applying the logical contraposition, we thus obtain $\Y \centernot\ind^{pw} \Z | \cT$. Thus there must be $Y \in \Y$ and $Z \in \Z$ such that $Y \centernot\ind Z | \cT$. By $\sigma$-Markovianity of $(\G,P_{\lX})$, $Y$ and $Z$ must be $\sigma$-connected by $\cT$. \\
Conversely assume that whenever $\Y$ and $\Z$ are $\sigma$-connected by a set $\cS \subset \cP$. Then there exist $Y \in \Y$ and $Z \in \Z$ that are $\sigma$-connected by $\mathcal{T} = \bigcup_{\W \in \cS} \W$. By $\sigma$-faithfulness of $(\G,P_{\lX})$ it follows that $Y \centernot\ind Z | \cT$ and thus $\Y \centernot\ind \Z | \cT $. Thus $(\co(\G,\cP),P_{\lX})$ is $\sigma$-faithful.

\end{proof}

\begin{proof}[Proof of Theorem \ref{lem.connectivity-criterion}]
Let $\W,\Y \in \cP$ and assume that $\Pi = (\Pi(1),\e_1,\Pi(2),\e_2,\dots,\e_{n-1},\Pi(n))$ is a path on $\co(\G,\cP)$ with $\Pi(1) = \W, \ \Pi(n)= \Y$ that is $\sigma$-unblocked by $\cS \subset \cP$. Let $\cT = \bigcup_{\Z \in \cS} \Z$. According to Corollary \ref{cor.macro-micro-faithful2}, we need to construct a path $\pi$ that coarsens to $\Pi$ and that is $\sigma$-unblocked by $\cT$. Consider the edge $\e_1$ and choose $e_1 \in \mic(\e_1)$ whose nodes we will immediately denote by $\pi(1) \in \Pi(1)$ and $\pi(2) \in \Pi(2)$, i.e. $e_1 = (\pi(1),\pi(2))$. For $i =2,\dots,n-1$, we proceed inductively as follows. Assume that we have already defined a path $(\pi(1),e_1\dots,e_{s-1},\pi(s))$ such that $\pi(s-1)\in \Pi(i-1), \ \pi(s) \in \Pi(i)$ and $e_{s-1} \in \mic(\e_{i-1})$. By assumption (iii) of the theorem, we can find a node $Y \in \mathrm{bd}_{\e_{i}}(\Pi(i)) \subset \Pi(i)$ such that $\scc(\pi(s)) = \scc(Y)$. If $\e_{i-1}$ is left-directed and $\e_i$ is left-directed or bidirected, choose a left-directed path $\xi(i)= (\xi(1) =\pi(s),\tilde{e}_1,\dots,\tilde{e}_m, \xi(m)=Y)$, in all other cases, choose a right-directed path $\xi_i= (\xi_i(1) =\pi(s),\tilde{e}_1,\dots,\tilde{e}_m, \xi_i(m)=Y)$. Note that by condition (ii), all nodes of $\xi_i$ must remain in $\Pi(i)$. Concatenate $\xi_i$ with $\pi$, i.e. set $e_{s-1+j} = \tilde{e}_j, \ \pi(s-1+j) = \xi_i(j), \ j=1,\dots, m$. Finally since $Y \in \mathrm{bd}_{\e_{i}}(\Pi(i))$ we can choose an edge $e_{s+m} \in \mic(\e_i)$ connecting $Y = \pi(s+m)$ to some $\pi(s+m+1) \in \Pi(i+1)$. If $i+1 = n$, we have finished the construction of our micro-path $\pi$ and by construction $\co(\pi) = \Pi$. We need to show now that $\cT$ $\sigma$-unblocks $\pi$. There are different cases to check.
\begin{itemize}
\item Assume that $\pi(1) \in \cT$ (respectively $\pi(n) \in \cT$). In this case we must have $\Pi(1) \in \cS$ (or $\Pi(n) \in \cS$) which would $\sigma$-block $\Pi$, contrary to our assumption. Thus $\pi(1), \pi(n) \notin \cT$.
\item Assume that $\pi(k)$ is a collider on $\pi$ for some $1 <k<\mathrm{len}(\pi)$ and that $\pi(k) \in \Pi(i)$. 
\begin{itemize}
\item The first case to discuss here is $|\Pi(i)| = 1$, i.e. $\Pi(i) = \{ \pi(k) \}$. In this case either $\Pi(i) \in \cS$ in which case $\pi(k) \in \cT$ or $\Pi(i)$ must have a proper descendant $\mathrm{S} \in \cS$. Then similar to the above construction of $\pi$, using (ii) and (iii) we can also construct a descending path $\pi(k) \to \dots \to S$ for some $S \in \cS$. Thus $\pi(k)$ has a descendant in $\cT$ and again the collider $\pi(k)$ is $\sigma$-unblocked.
\item The second case is $|\Pi(i)| > 1$. Because of our choice of the internal path $\xi_i$ as directed, $\pi(k) \in \mathrm{bd}_{\e_{i}}$ or $\pi(k) \in \mathrm{bd}_{\e_{i-1}}$. We will only discuss the first case $\pi(k) \in \mathrm{bd}_{\e_{i}}$ as the second one is completely analogous. If $\pi(k) \in \mathrm{bd}_{\e_{i}}$, then $\pi(k+1) \in \Pi(i+1)$ and the edge $e_k = (\pi(k),\pi(k+1))$ must be left- or bidirected as $\pi(k)$ is a collider. Again because of the way we chose $\xi_i$, the unique edge on $\pi$ in $\mic(\e_{i-1})$ must be right- or bidirected. Thus both $\e_{i-1}$ and $\e_{i}$ must have an arrowhead towards $\Pi(i)$, i.e $\Pi(i)$ is a collider on $\Pi$. Thus $\Pi(i)$ must have a descendant in $\cS$. Suppose first that this descendant is $\Pi(i)$ itself. Then the collider $\pi(k)$ is $\sigma$-unblocked as it is contained in $\cT$. The other nodes on $\pi$ that are part of $\Pi(i)$ are non-collider but do not point to neighbors on $\pi$ that are part of a different strongly connected component. Thus by condition (3) in the definition of $\sigma$-separation (Definition \ref{def.sigma-sep}) including them in $\cT$ does not $\sigma$-block $\pi$. Next, suppose that the descendant of $\Pi(i)$ in $\cS$ is a proper descendant. Once again, using (ii) and (iii) we can construct a descending path $\pi(k) \to \dots \to S$ for some $S \in \cS$ so that the collider $\pi(k)$ is unblocked.
\end{itemize} 
\item Assume that $\pi(k)$ is a non-collider on $\pi$ for some $1 <k<\mathrm{len}(\pi)$ and that $\pi(k) \in \Pi(i)$. As $\cS$ $\sigma$-unblocks $\Pi$, we must be in one of the following situations. Either (I) $\Pi(i) \notin \cS$ or (II) $\Pi(i) \in \cS$ but if $\e_{i-1} = \Pi(i-1) \leftarrow \Pi(i)$ or $\e_i = \Pi(i) \to \Pi(i+1)$ then $\scc(\Pi(i-1)) = \scc(\Pi(i))$, respectively $\scc(\Pi(i)) = \scc(\Pi(i+1))$.
\begin{itemize}
\item[(I)] In this case, $\pi(k) \notin \cT$ thus the non-collider $\pi(k)$ is $\sigma$-unblocked by $\cT$.
\item[(II)] If $\Pi(i) \in \cS$ suppose that $\e_i = \Pi(i) \to \Pi(i+1)$ has a tail at $\Pi(i)$. As stated above, the fact that $\cS$ $\sigma$-unblocks $\Pi$ means that we must have $\scc(\Pi(i+1)) = \scc(\Pi(i))$. If $\pi(k+1)$ is also an element of $\Pi(i)$, then by construction of $\pi$, it has a right-directed edge $\pi(k) \to \pi(k+1)$ and $\scc(\pi(k)) = \scc(\pi(k+1))$. Thus, even though $\pi(k) \in \cT$, it is still $\sigma$-unblocked by $\cT$. Thus we can assume that $\pi(k+1) \in \Pi(i+1)$ which means in particular that $\pi(k) \in \mathrm{bd}_{\e_i}(\Pi(i))$ and $\pi(k+1) \in \mathrm{bd}_{\e_i}(\Pi(i+1))$.  As $\scc(\Pi(i+1)) = \scc(\Pi(i))$, there exists a directed path $\Gamma$ on $\co(\G,\cP)$ starting at $\Pi(i+1)$ and ending at $\Pi(i)$. As with the construction of $\pi$ above, because of the boundary connection condition (iii), we can once again construct a micro-path $\gamma$ from $\pi(k+1)$ to $\pi(k)$ so that $\scc(\pi(k)) = \scc(\pi(k+1))$. So, once again even though $\pi(k) \in \cT$, it is still $\sigma$-unblocked by $\cT$. The final case is that $\e_i$ does not have a tail at $\Pi(i)$ which means that $\e_{i-1} = \Pi(i-1) \leftarrow \Pi(i)$ must have one. The argument that  $\scc(\pi(k)) = \scc(\pi(k-1))$ is then completely parallel to the discussion for tailed $\e_i$, taking into account that the path segment $\xi_i$ of $\pi$ that is internal to $\Pi(i)$ is left-directed (or trivial) by construction. Therefore also in this case, even though $\pi(k) \in \cT$, it is still $\sigma$-unblocked by $\cT$.
\end{itemize}
\end{itemize}
We have shown above that every collider of $\pi$ has a descendant in $\cT$ and that every non-collider is either not part of $\cT$ or points exclusively to neighbors in the same connected component. In summary  $\cT$ $\sigma$-unblocks $\pi$, so $\sigma$-faithfulness is proven. \\
To show that $\co(\G,\cP)$ is acyclic, assume that there exists a right-directed cycle $\Pi$ with $ \Pi(1)  =\Pi(\mathrm{len}(\Pi))$. W.l.o.g., we can assume that $\Pi$ is irreducible. Let $\cP'$ be the partition of $\cG$ into strongly connected components and $\mathcal{H} = \co(\G,\cP')$ which is always acyclic. Then condition (iii) implies that $\Pi$ induces a right-directed walk $\Gamma_0$ on $\mathcal{H}$ with no repeating middle vertices such that $\Gamma_0(1), \Gamma_0(\mathrm{len}(\Gamma_0)) \subset \Pi(1)$. 
If $\Gamma_0(1) = \Gamma_0(\mathrm{len}(\Gamma_0))$ we have found a cycle in $\mathcal{H}$ and thus a contradiction. If not, we can again use condition (iii) to construct a right-directed walk $\Gamma_0'$ with $\Gamma_0'(1) = \Gamma_0(\mathrm{len}(\Gamma_0))$ and $\Gamma_0'(\mathrm{len}(\Gamma_0')) \subset \Pi(1)$. Concatenating $\Gamma_0$ and $\Gamma_0$ to $\Gamma_1 = \Gamma_0 \circ \Gamma_0'$, we have found two walks now that start in the same strongly connected component $\Gamma_0(1)$ and end in $\cP'\cap 2^{\Pi(1)}$.  Again if $\Gamma_1(\mathrm{len}(\Gamma_1)) = \Gamma(1)$ we are done, otherwise we continue to construct walks $\Gamma_k$ on $\cH$ in this manner. Since the set $\cP'\cap 2^{\Pi(1)}$ is finite, at some point the condition  $\Gamma_k(\mathrm{len}(\Gamma_k)) = \Gamma(k)$ must be be met and we arrive at a contradiction.


\end{proof}

\begin{proof}[Proof of Theorem \ref{lem.connectivity-criterion2}]
Let $\W,\Y \in \cP$ and assume that $\Pi = (\Pi(1),\e_1,\Pi(2),\e_2,\dots,\e_{n-1},\Pi(n))$ is a path on $\co(\G,\cP)$ with $\Pi(1) = \W, \ \Pi(n)= \Y$ that is $\sigma$-unblocked by $\cS \subset \cP$. Let $\cT = \bigcup_{\Z \in \cS} \Z$. According to Corollary \ref{cor.macro-micro-faithful2}, we need to construct a path $\pi$ that coarsens to $\Pi$ and that is $\sigma$-unblocked by $\cT$. We proceed as in the proof of Lemma \ref{lem.connectivity-criterion} and for non-colliders $(\e_i,\Pi(i+1),\e_{i+1})$, the construction of $\pi$ completely parallels that proof. For a collider $(\e_i,\Pi(i+1),\e_{i+1})$, we simply can choose the micro path $\xi_i$ in the argument to be a micro-collider $(\xi_i(1), e^i_1,\xi(2),e^i_2,\xi(3))$ with $\xi_i(1) \in \Pi(i-1), \ \xi_i(2) \in \Pi(i), \ \xi_i(3) \in \Pi(i+1)$.
Showing that $\cT$ unblocks $\pi$ again completely parallels the the proof of Lemma \ref{lem.connectivity-criterion}, noting that if $\pi(k)$ is a non-collider on $\pi$ it must lie in a group $\Pi(i)$ that is a non-collider of $\Pi$. If $\pi(k)$ is a collider on $\pi$, then $\pi(k) \in \Pi(i)$ and $\Pi(i)$ must be a collider on $\Pi$. Since $\cS$ $\sigma$-unblocks $\Pi(i)$, we must either have $\Pi(i) \in \cS$ or $\Pi(i)$ must have a proper descendant in $\cS$, say $\W$. In the first case, it follows that $\pi(k) \in \cT$. In the second case, consider the descending path $\Gamma = (\Pi(i), \tilde{\e}_1, \dots, \tilde{\e}_s, \W)$. We can again construct a descending micro-path $\gamma$ starting at some $\gamma(1) \in \mathrm{bd}_{\tilde{\e}_1}(\Pi(i))$ to some $W \in \W$. Since the pair $\e_i,\tilde{\e}_1$ is an (almost) mediator, we can use condition (iii-a) to find a directed micro-path from $\pi(k)$ to $\gamma(1)$ and hence to  $W$. Thus, $\pi(k)$ has a descendant in $\cT$ as desired.
\end{proof}

\end{appendices}

\end{document}